\documentclass[]{gIPE2e}
\usepackage{amssymb}
\usepackage{amsmath}
\usepackage{fancyhdr}
\usepackage{multirow}

\begin{document}
\doi{10.1080/1741597YYxxxxxxxx}
 \issn{1741-5985}
\issnp{1741-5977}

\jvol{00} \jnum{00} \jyear{2012} \jmonth{August}

\markboth{J. Su et al.}{Inverse Problems in Science and Engineering}


\title{{\itshape Optical imaging of phantoms from real data by an approximately globally convergent inverse algorithm}}


\author{Jianzhong Su$^{\rm a}$, %
Michael V. Klibanov$^{\rm b}$, %
Yueming Liu$^{\rm a}$, %
Zhijin Lin$^{\rm c}$, %
Natee Pantong$^{\rm d}$ and %
Hanli Liu$^{\rm c}$\\\vspace{6pt} %
$^{\rm a}${\em{Department of Mathematics, University of Texas at Arlington, Arlington, TX 76019, USA}}; %
$^{\rm b}${\em{Department of Mathematics and Statistics, University of North Carolina at Charlotte, Charlotte, NC 28223, USA}}%
$^{\rm c}${\em{Department of Bioengineering, University of Texas at Arlington, Arlington, TX 76019, USA}}%
$^{\rm d}${\em{Department of Mathematics and Computer Science, Royal Thai Air Force Academy, Bangkok, Thailand}}%
}

\maketitle

\begin{abstract}
A numerical method for an inverse problem for an elliptic equation
with the running source at multiple positions is presented. This
algorithm does not rely on a good first guess for the solution. The
so-called \textquotedblleft approximate global convergence" property
of this method is shown here. The performance of the algorithm is
verified on real data for Diffusion Optical Tomography. Direct
applications are in near-infrared laser imaging technology for
stroke detection in brains of small animals.\bigskip

\begin{keywords}Approximate Global Convergence Property; Inverse Problem; Diffusion
Optical Tomography; Real Data
\end{keywords}
\begin{classcode}65B21; 65D10; 65F10\end{classcode}\bigskip

\end{abstract}

\section{Introduction}

We consider a Coefficient Inverse Problem (CIP) for a partial
differential equation (PDE) - the diffusion model with the unknown
potential. The boundary data for this CIP, which model measurements,
are originated by a point source running along a part of a straight
line. This PDE governs light propagation in a diffusive medium, such
as, e.g. biological tissue, smog, etc.. Thus, our CIP is one of
problems of Diffusion Optical Tomography (DOT). We are interested in
applications of DOT to the detection of stroke in small animals
using measurements of near infrared light originated by lasers.
Hence, the above point source is the light source in our case. The
motivation of imaging of small animals comes from the idea that it
might be a model case for the future stroke detection in humans, at
least in brains of neonatals, via DOT. We apply our numerical method
to a set of real data for a phantom medium modeling the mouse's
brain. Although this algorithm was developed in earlier publications
\cite{bib18,bib22,bib23,bib24} of this group, its experimental
verification is new and is the main contribution of this paper.

As to our numerical method, we introduce a new concept of the
\textquotedblleft approximate global convergence" property. In the
previous publication \cite{bib18} of this group about this method
the approximate global convergence property was established in the
continuos case. Compared with \cite{bib18}, the main new analytical
result here is that we establish this property for the more
realistic discrete case. In the convergence
analysis of \cite{bib18} the Schauder theorem \cite{bib16} was applied for $%
C^{2+\alpha }-$solutions of certain elliptic equations arising in
our method. Now, however, since we consider the discrete case, we
use the Lax-Milgram theorem. Here and below $C^{k+\alpha }$ are
H\"{o}lder spaces \cite{bib16}, where $k\geq 0$ is an integer and
$\alpha \in \left( 0,1\right) $.

CIPs are both nonlinear and ill-posed. These two factors cause very
serious challenges in their numerical treatments. Indeed,
corresponding least squares Tikhonov regularization functionals
usually suffer from multiple local minima and ravines. As a result,
conventional numerical methods for CIPs are locally convergent ones,
see, e.g. \cite{bib4} and references cited there. To have a
guaranteed convergence to a true solution, a locally convergent
algorithm should start from a point which is located in a small
neighborhood of this solution. However, it is a rare case in
applications when such a point is known. The main reason why our
method avoids local minima is that it uses the structure of the
underlying PDE operator and does not use a least squares functional.

Therefore,  it is important to develop such numerical methods for
CIPs, which would have a rigorous guarantee of providing a point in
a small neighborhood of that exact solution without any \emph{a
priori} knowledge of this neighborhood. It is well known, however,
that the goal of the development of such algorithms is a
substantially challenging one. Hence, it is unlikely that such
numerical methods would not rely on some approximations. This is the
reason why the notion of the approximate global convergence property
was introduced in the recent work \cite{bib20}; also see section
1.1.2 of the book \cite{bibBook} and subsection 4.1 below. The
verification of our approximately globally convergent numerical
method on computationally simulated data was done in \cite%
{bib18,bib22,bib23,bib24}. In this paper we make the next step: we
verify this method on real data. Regardless on some approximations
we have here, the \emph{main point} is that our numerical method
does not rely on any a priori knowledge of a small neighborhood of
the exact solution.

In \cite{bib7,bib9,bib10,bib19,bib20,bib201} a similar numerical
method for CIPs for a hyperbolic PDE was developed. These results
were summarized in the book \cite{bibBook}. In particular, it was
demonstrated numerically in
Test 5 of \cite{bib7}, section 8.4 of \cite{bib19}, on page 185 of \cite%
{bibBook} and in section 5.8.4 of \cite{bibBook} that the
approximately globally convergent method of
\cite{bib7,bib9,bib10,bib19,bib20,bib201} outperforms such locally
convergent algorithms as quasi-Newton method and gradient method.
The main difference of the technique of these publications with the
one of the current paper is in the truncation of a certain integral.
In \cite{bib7,bib9,bib10,bib19,bib20,bib201} that integral was
truncated at a high value $\overline{s}>0$ of the parameter $s>0$ of
the Laplace transform of the solution of the underlying hyperbolic
PDE. Indeed, in the hyperbolic case the truncated residual of that
integral, which we
call the \textquotedblleft tail function\textquotedblright , is $O\left( 1/%
\overline{s}\right) $, as $\overline{s}\rightarrow \infty $, i.e.
this residual is automatically small in this case. Being different
from the hyperbolic case, in the elliptic PDE with the source
running along a straight line, $\overline{s}>0$ represents the
distance between the source and the origin. Because of this, the
tail function is not automatically small in our case. Thus, a
special effort to ensure this smallness was undertaken in
\cite{bib18}, and is repeated in the current publication. It was
shown numerically in \cite{bib18} that this special effort indeed
improves the quality of the reconstruction, compare Figures 2b and
2c in \cite{bib18}.

Our real data were collected at the boundary of a 2-D cross-section
of the 3-D domain of interest, see section 6. Hence, we have imaged
this cross-section only and have ignored the dependence on the third
variable. In addition, our theory requires the use of many sources
placed along a straight line. However, limitations of our
experimental device allow us to use only three sources on this line.
We believe that the accuracy of our results (Table 8.2) confirms the
validity of both these assumptions as well as manifests a well known
fact that there are always some discrepancies between the analysis
and computational studies of numerical methods.

Because of the above mentioned substantial challenges, the topic of
the development of non-locally convergent numerical methods for CIPs
is currently in its infancy. As to such methods for CIPs for
elliptic PDEs, we refer to, e.g. publications
\cite{bib2,bib12,bib17,bib21,bib30,bib31} and references cited
there. These publications are concerned with the
Dirichlet-to-Neumann map (DN). We are not using the DN\ here.

In a general setting, an ill-posed problem is a problem of solving
an equation $F\left( x\right) =y$ \ with a compact operator $F$.
Here $x\in B_{1}$, $y\in B_{2},$ where $B_{1}$ and $B_{2}$ are
certain Banach spaces. Naturally, the operator $F$ varies from one
problem to another one. It is well known that the existence theorem
for such an equation is very tough to prove, since the range of any
compact operator is `too narrow' in certain sense, see, e.g. Theorem
1.2 in \cite{bibBook}. Therefore, by one of fundamental concepts of
the theory of Ill-Posed Problems, one should assume the existence of
an exact solution of such a problem for the case of an
\textquotedblleft ideal" noiseless data \cite{bib5,bibBook,bib25}.
Although this solution is never known in practice and is never
achieved in practice (because of the noise in the real data), the
regularization theory says that one needs to construct a good
approximation for it. We assume that our CIP has an exact solution
for noiseless data and also assume that this solution is unique.

The rest of this paper is arranged as follows. In section 2 we pose
both forward and inverse problems and study some properties of the
solution of the forward problem. In section 3 we present our
numerical method. In section 4 we conduct the convergence analysis.
In section 5 we discuss the numerical implementation of our method.
In section 6 we describe the experiment. In section 7 we outline our
procedure of processing of real data. In section 8 we present
reconstruction results. We briefly summarize results in section 9.

\section{Statement of the Problem}

\subsection{The Inverse Problem}

Below \thinspace $\mathbf{x}=\left( x,z\right) \in \mathbb{R}^{2}$ and $%
\Omega \subset $ $\mathbb{R}^{2}$ is a convex bounded domain. The
boundary of this domain $\partial \Omega \in C^{3}$ in our analysis.
In numerical studies $\partial \Omega $ is piecewise smooth.
Consider the following elliptic equation in $\mathbb{R}^{2}$ with
the solution vanishing at infinity,
\begin{eqnarray}
\Delta u-a\left( \mathbf{x}\right) u &=&-\delta \left( \mathbf{x}-\mathbf{x}%
_{0}\right) ,\mathbf{x},\mathbf{x}_{0}\in \mathbb{R}^{2},  \label{2.1} \\
\underset{\left\vert \mathbf{x}\right\vert \rightarrow \infty }{\lim }%
u\left( \mathbf{x},\mathbf{x}_{0}\right)  &=&0.  \label{2.2}
\end{eqnarray}

\textbf{Inverse Problem.} \emph{Let }$k=const.>0$\emph{. Suppose that in (%
\ref{2.1}) the coefficient }$a\left( \mathbf{x}\right) $\emph{\
satisfies the following conditions }
\begin{equation}
a\in C^{1}\left( \mathbb{R}^{2}\right) ,\text{ }a\left(
\mathbf{x}\right)
\geq k^{2}\text{ and }a\left( \mathbf{x}\right) =k^{2}\text{ for }\mathbf{x}%
\in \mathbb{R}^{2}\diagdown \Omega .  \label{2.3}
\end{equation}%
\emph{Let\ }$L\subset \left( \mathbb{R}^{2}\diagdown \overline{\Omega }%
\right) $\emph{\ be a straight line and }$\Gamma \subset L$\emph{\
be an
unbounded and connected subset of }$L$.\emph{\ Determine the function }$%
a\left( \mathbf{x}\right) $\emph{\ inside of the domain }$\Omega
,$\emph{\ assuming that the constant }$k$\emph{\ is given and also
that the following function }$\varphi \left(
\mathbf{x},\mathbf{x}_{0}\right) $\emph{\ is given }
\begin{equation}
u\left( \mathbf{x,x}_{0}\right) =\varphi \left(
\mathbf{x,x}_{0}\right) ,\forall \left( \mathbf{x,x}_{0}\right) \in
\partial \Omega \times \Gamma . \label{2.4}
\end{equation}

We are unaware about a uniqueness theorem for this inverse problem.
Nevertheless, because of the above application, it is reasonable to
we develop a numerical method. Thus, we assume that uniqueness for
this problem holds. Our numerical studies of both the past
\cite{bib18,bib22,bib23,bib24} and the current publication indicate
that a certain uniqueness theorem might be established.

We assume that sources $\left\{ \mathbf{x}_{0}\right\} $ are located
outside of the domain of interest $\Omega $ because this is the case
of our measurements and because we do not want to work with
singularities in our numerical method.
The CIP (\ref{2.1})-(\ref{2.4}) has an application in imaging using
light propagation in a diffuse medium, such as biological tissues.
Since the modulated frequency equals zero in our case, then this is
the so-called continuous-wave (CW) light. The coefficient $a\left(
\mathbf{x}\right) =3\left( \mu _{s}^{\prime }\mu _{a}\right) \left(
\mathbf{x}\right) ,$ where $\mu _{s}^{\prime }\left(
\mathbf{x}\right) $ is the reduced scattering coefficient and $\mu
_{a}\left( \mathbf{x}\right) $ is the absorption coefficient of the
medium \cite{bib3,bib4}. In the case of our particular interest in
stroke detections in brains of small animals, the area of an early
stroke can be modeled as a small sharp inclusion in an otherwise
slowly fluctuating background. Usually the inclusion/ background contrast $%
a_{incl}/a_{b}\geq 2.$ Therefore our focus is on the reconstruction
of both
locations of sharp small inclusions and the values of the coefficient $%
a\left( \mathbf{x}\right) $ inside of them, rather than on the
reconstruction of slow changing background functions.

\subsection{Some Properties of the Solution of the Forward Problem (\protect
\ref{2.1}), (\protect\ref{2.2})}

\subsubsection{Existence and uniqueness}

First, we state the existence and uniqueness of the solution of the
forward problem (\ref{2.1}), (\ref{2.2}). For brevity we consider
only the case $\mathbf{x}_{0}\notin \overline{\Omega },$ since this
is the case of our Inverse Problem. Let $K_{p}\left( z\right) ,z\in
\mathbb{R},p\geq 0$ be
the Macdonald function. It is well known \cite{bib1} that for $y\in \mathbb{R%
}$
\begin{equation}
K_{p}\left( y\right) =\frac{\sqrt{\pi }}{\sqrt{2y}}e^{-y}\left(
1+O\left( \frac{1}{y}\right) \right) ,y\rightarrow \infty .
\label{2.5}
\end{equation}

\begin{theorem}
\label{th:2.1} Let $\Omega \subset \mathbb{R}^{2}$\emph{\ be the
above
bounded domain. Assume that the coefficient }$a\left( \mathbf{x}\right) $%
\emph{\ satisfies conditions (\ref{2.3}). Then for each source position }$%
x_{0}\in \mathbb{R}^{2}\diagdown \overline{\Omega }$\emph{\ there
exists
unique solution }$u\left( \mathbf{x,x}_{0}\right) $\emph{\ of the problem (%
\ref{2.1}), (\ref{2.2}) such that  }%
\begin{equation}
u\left( \mathbf{x,x}_{0}\right) =\frac{1}{2\pi }K_{0}\left(
k\left\vert
\mathbf{x-x}_{0}\right\vert \right) +\widehat{u}\left( \mathbf{x,x}%
_{0}\right) :=u_{0}\left( \mathbf{x-x}_{0}\right) +\widehat{u}\left( \mathbf{%
x,x}_{0}\right) ,  \label{2.6}
\end{equation}%
\emph{where the function }$u_{0}$\emph{\ is the fundamental solution
of equation (\ref{2.1}) with }$a\left( \mathbf{x}\right) \equiv
k^{2},$\emph{\ the function }$\widehat{u}$\emph{\ satisfies
(\ref{2.2}), }$\widehat{u}\in H^{2}\left( \mathbb{R}^{2}\right)
$\emph{\ and }$\widehat{u}\in C^{2+\alpha
}\left( \mathbb{R}^{2}\right) .$\emph{\ In addition, }$u\left( \mathbf{x,x}%
_{0}\right) >0,\forall x\in \overline{\Omega }.$
\end{theorem}

A similar result for the 3-D case was proven in \cite{bibBook}, see
Theorem 2.7.2 in this reference. Hence, we leave out the proof of
Theorem 2.1 for brevity.

\subsubsection{The asymptotic behavior at $\left| \mathbf{x}_{0}\right|
\rightarrow \infty $}

It follows from (\ref{2.5}) and (\ref{2.6}) that the asymptotic
behavior of the function $u_{0}\left( \mathbf{x-x}_{0}\right) $ is
\begin{eqnarray*}
u_{0}\left( \mathbf{x-x}_{0}\right) &=&w_{0}\left( \left\vert \mathbf{x}%
_{0}\right\vert \right) \left( 1+O\left( \frac{1}{\left\vert \mathbf{x}%
_{0}\right\vert }\right) \right) ,\left\vert
\mathbf{x}_{0}\right\vert
\rightarrow \infty , \\
w_{0}\left( \left\vert \mathbf{x}_{0}\right\vert \right) &=&\frac{%
e^{-k\left\vert \mathbf{x}_{0}\right\vert }}{2\sqrt{2\pi \left\vert \mathbf{x%
}_{0}\right\vert }}.
\end{eqnarray*}%
Denote $b\left( \mathbf{x}\right) =a\left( \mathbf{x}\right)
-k^{2}$. Then
by (\ref{2.3}) $b\left( \mathbf{x}\right) =0$ for $\mathbf{x}\in \mathbb{R}%
^{2}\diagdown \Omega $. Let $M_{1}$ be a positive constant. Denote
\begin{equation*}
B\left( M_{1}\right) =\left\{ b\in C^{1}\left( \mathbb{R}^{2}\right)
:\left\Vert b\right\Vert _{C^{1}\left( \overline{\Omega }\right)
}\leq M_{1},b\left( \mathbf{x}\right) \geq 0,b\left(
\mathbf{x}\right) =0\text{ for }\mathbf{x}\in
\mathbb{R}^{2}\diagdown \Omega \right\} .
\end{equation*}%
Also, let the function $p_{\infty }\left( \mathbf{x}\right) $
satisfies conditions
\begin{equation}
p_{\infty }\left( \mathbf{x}\right) \in C^{2+\alpha }\left(
\left\vert \mathbf{x}\right\vert \leq R\right) ,\forall
R>0,p_{\infty }\in H^{2}\left( \mathbb{R}^{2}\right) .  \label{2.25}
\end{equation}%
and be the solution of the following problem
\begin{eqnarray}
\Delta p_{\infty }-k^{2}p_{\infty }-b\left( \mathbf{x}\right)
p_{\infty }
&=&b\left( \mathbf{x}\right) ,\mathbf{x}\in \mathbb{R}^{2},  \label{2.26} \\
\underset{\left\vert x\right\vert \rightarrow \infty }{\lim
}p_{\infty }\left( \mathbf{x}\right) &=&0.  \label{2.27}
\end{eqnarray}%
The uniqueness and existence of the solution of the problem (\ref{2.25})-(%
\ref{2.27}) are similar to these of Theorem \ref{th:2.1}.

\begin{lemma}
\label{Le:2.1} $1+p_{\infty }>0.$
\end{lemma}

\begin{proof}
Let $\widetilde{p}=1+p_{\infty }.$ Then
\begin{equation}
\Delta \widetilde{p}-k^{2}\widetilde{p}-b\left( \mathbf{x}\right) \widetilde{%
p}=-k^{2},\text{ }\lim_{\left\vert \mathbf{x}\right\vert \rightarrow \infty }%
\widetilde{p}\left( \mathbf{x}\right) =1.  \label{2.28}
\end{equation}%
Consider a sufficiently large number $R>0$ such that
$\widetilde{p}\left(
\mathbf{x}\right) \geq 1/2$ for $\mathbf{x\in }\left\{ \left\vert \mathbf{x}%
\right\vert \geq R\right\} .$ Then the maximum principle
\cite{bib16} applied to equation (\ref{2.28}) for $\mathbf{x}\in
\left\{ \left\vert
\mathbf{x}\right\vert <R\right\} $ shows that $\widetilde{p}\left( \mathbf{x}%
\right) >0$ in $\left\{ \left\vert \mathbf{x}\right\vert <R\right\}
$.
\end{proof}

\begin{lemma}\label{Le:2.2}\cite{bib18}
\emph{Let the function }$b\in B\left( M_{1}\right) $\emph{. Then
there exists a constant }$M_{2}\left( M_{1},\Omega \right)
>0$\emph{\ such that}
\begin{equation*}
\left\Vert \ln u\left( \mathbf{x,x}_{0}\right) -\ln w_{0}\left(
\left\vert
\mathbf{x}_{0}\right\vert \right) -\ln \left( 1+p_{\infty }\left( \mathbf{x}%
\right) \right) \right\Vert _{C^{2}\left( \overline{\Omega }\right)
}\leq
\frac{M_{2}\left( M_{1},\Omega \right) }{\left\vert \mathbf{x}%
_{0}\right\vert },
\end{equation*}
\begin{equation*}
\mathbf{x}_{0}\in \left\{ \left\vert \mathbf{x}_{0}\right\vert
>1\right\} \cap \left( \mathbb{R}^{2}\diagdown \overline{\Omega
}\right) ,\mathbf{x}\in \overline{\Omega }.
\end{equation*}
\emph{\ }
\end{lemma}

\section{The Numerical Method}

Since we can put the origin on\ the straight line $L$, if necessary,
then without any loss of the generality we can set $s:=\left\vert
\mathbf{x}_{0}\right\vert$, assuming that only the parameter $s$
changes when the source $x_{0}$ runs along $\Gamma \subset L$.
Denote $u\left( \mathbf{x},s\right) :=u\left( \mathbf{x}%
,\mathbf{x}_{0}\right) ,\mathbf{x}\in \Omega ,\mathbf{x}_{0}\in
\Gamma .$
Since $\Gamma \cap \overline{\Omega }=\varnothing $ and the point source $%
\mathbf{x}_{0}\in \Gamma $, then $\mathbf{x}_{0}\notin
\overline{\Omega }$. Since by Theorem \ref{th:2.1} $u\left(
\mathbf{x},s\right) >0,\forall
\mathbf{x}\in \overline{\Omega }$, then we can consider the function $%
\widetilde{w}\left( \mathbf{x},s\right) =\ln u\left(
\mathbf{x},s\right) $ for $\mathbf{x}\in \Omega $. We obtain from
(\ref{2.1}) and (\ref{2.4})
\begin{eqnarray}
\Delta \widetilde{w}+\left\vert \nabla \widetilde{w}\right\vert ^{2}
&=&a\left( \mathbf{x}\right) \text{ in }\Omega ,  \label{3.1} \\
\widetilde{w}\left( \mathbf{x},s\right)  &=&\varphi _{1}\left( \mathbf{x}%
,s\right) ,\forall \left( \mathbf{x},s\right) \in \partial \Omega \times %
\left[ \underline{s},\bar{s}\right] ,  \label{3.2}
\end{eqnarray}%
where $\varphi _{1}=\ln \varphi $ and $\underline{s},\bar{s}$ are
two positive numbers, which should be chosen in numerical
experiments.

\subsection{The integral differential equation}

We now \textquotedblleft eliminate\textquotedblright\ the coefficient $%
a\left( \mathbf{x}\right) $ from equation (\ref{3.1}) via the
differentiation with respect to $s$. However, to make sure that the
resulting the so-called \textquotedblleft tail
function\textquotedblright\ is small, we use the above mentioned
(Introduction) special effort of the paper \cite{bib18}. Namely, we
divide (\ref{3.1}) by $s^{p},p>0.$ In
principle, any number $p>0$ can be used. But since in computations we took $%
p=2,$ both here and in \cite{bib18}, then we use below only $p=2,$
for the
sake of definiteness. Denote $w\left( \mathbf{x},s\right) =$ $\widetilde{w}%
\left( \mathbf{x},s\right) /s^{2}$. By (\ref{3.1}) equation for the
function
$w\left( \mathbf{x},s\right) $ is%
\begin{equation}
\Delta w+s^{2}\left\vert \nabla w\right\vert ^{2}=\frac{a\left( \mathbf{x}%
\right) }{s^{2}}.  \label{3.21}
\end{equation}%
Next, let $q\left( \mathbf{x},s\right) :=\partial _{s}w\left( \mathbf{x}%
,s\right) =\partial _{s}\left( s^{-2}\ln u\left( \mathbf{x},s\right)
\right) ,$ for $s\in \left[ \underline{s},\overline{s}\right] $.
Then
\begin{equation}
w\left( \mathbf{x},s\right) =-\int\limits_{s}^{\bar{s}}q\left( \mathbf{x}%
,\tau \right) d\tau +T\left( \mathbf{x}\right) ,\mathbf{x}\in \Omega ,s\in %
\left[ \underline{s},\overline{s}\right] ,  \label{3.3}
\end{equation}%
where $T\left( \mathbf{x}\right) $ is the so-called
\textquotedblleft tail function\textquotedblright . The exact
expression for this function is of course $T\left( \mathbf{x}\right)
=w\left( \mathbf{x},\overline{s}\right) $. However, since the
function $w\left( \mathbf{x},\overline{s}\right) $ is unknown, we
will use an approximation for the tail function, see subsection 5.2
as well as \cite{bib18}. By the Tikhonov concept for ill-posed
problems \cite{bib5,bibBook,bib25}, one should have some a priori
information about the solution of an ill-posed problem. Thus, we can
assume the knowledge of a constant $M_{1}>0$ such that the function
$a\left( \mathbf{x}\right) -k^{2}=b\left( \mathbf{x}\right) \in
B\left( M_{1}\right) $. Hence, it follows from Lemma \ref{Le:2.2}
that
\begin{eqnarray}
&&T\left( \mathbf{x},\overline{s}\right) =\frac{\ln w_{0}\left( \overline{s}%
\right) }{\overline{s}^{2}}+\frac{\ln \left( 1+p_{\infty }\left( \mathbf{x}%
\right) \right) }{\overline{s}^{2}}+\frac{g\left( \mathbf{x},\overline{s}%
\right) }{\overline{s}^{3}},\text{ }\mathbf{x}\in \Omega ,\text{
}\forall
\overline{s}>1,  \label{3.4} \\
&&\left\Vert g\left( \mathbf{x},\overline{s}\right) \right\Vert
_{C^{1}\left( \overline{\Omega }\right) }\leq M_{2}\left(
M_{1},\Omega \right) ,\text{ }\forall \overline{s}>1,\text{ }\forall
b\in B\left( M_{1}\right) ,  \notag
\end{eqnarray}%
where the number $M_{2}\left( M_{1},\Omega \right) $ is independent on $%
\overline{s}.$ Differentiating (\ref{3.21}) with respect to $s,$ we
obtain the following nonlinear integral differential equation for
the function $q$ \cite{bib18}
\begin{equation}
\Delta q-\frac{2}{s}\int\limits_{s}^{\bar{s}}\Delta q\left(
\mathbf{x},\tau \right) d\tau -2s^{2}\nabla
q\int\limits_{s}^{\bar{s}}\triangledown q\left( \mathbf{x},\tau
\right) d\tau   \label{3.5}
\end{equation}%
\begin{equation*}
+4s\left( -\int\limits_{s}^{\bar{s}}\triangledown q\left(
\mathbf{x},\tau
\right) d\tau +\triangledown T\right) ^{2}+2s^{2}\triangledown T\nabla q=-%
\frac{2}{s}\Delta T.
\end{equation*}%
In addition, (\ref{3.2}) implies that
\begin{eqnarray}
q\left( \mathbf{x},s\right)  &=&\psi \left( \mathbf{x},s\right) ,\text{ }%
\forall \left( \mathbf{x},s\right) \in \partial \Omega \times \left[
\underline{s},\bar{s}\right] ,  \label{3.6} \\
\psi \left( \mathbf{x},s\right)  &=&\partial _{s}\left( s^{-2}\ln
\varphi \left( \mathbf{x},s\right) \right) .  \label{3.7}
\end{eqnarray}

If we approximate well both functions $q$ and $T$ together with
their derivatives up to the second order, then we can also
approximate well the target coefficient $a\left( \mathbf{x}\right) $
via backwards calculations, see (\ref{5.1}). Therefore, the main
questions now is:\emph{\ How to
approximate well both functions }$q$\emph{\ and }$T$\emph{\ using (\ref{3.3}%
)-(\ref{3.7})? }

\subsection{Layer stripping with respect to the source position}

In this subsection we present a layer stripping procedure with
respect to $s$ for approximating the function $q$, assuming that the
function $T$ is known. Usually the layer stripping procedure is
applied with respect to a spatial variable. However, sometimes the
presence of a differential operator with respect to this variable in
the underlying PDE results in the instability, since computing
derivatives is an unstable procedure. The reason why our layer
stripping procedure is stable is that we do not have a differential
operator with respect to $s$ in our PDE.

We approximate the function $q\left( \mathbf{x},s\right) $ as a
piecewise constant function with respect to the source position $s$.
We assume that
there exists a partition $\underline{s}=s_{N}<s_{N-1}<\ldots <s_{1}<s_{0}=%
\bar{s},s_{i-1}-s_{i}=h$ of the interval $\left[
\underline{s},\bar{s}\right] $ with a sufficiently small step size
$h$ such that
\begin{equation}
q\left( \mathbf{x},s\right) =q_{n}\left( \mathbf{x}\right) \text{ for }s\in %
\left[ s_{n},s_{n-1}\right) ,n\geq 1;\text{ }q_{0}:\equiv 0.
\label{3.8}
\end{equation}%
Hence,
\begin{equation*}
\int\limits_{s}^{\bar{s}}q\left( \mathbf{x},s\right) ds=\left(
s_{n-1}-s\right) q_{n}\left( \mathbf{x}\right)
+h\sum\limits_{j=0}^{n-1}q_{j}\left( \mathbf{x}\right) .
\end{equation*}%
Let $\psi _{n}\left( \mathbf{x}\right) $ be the average of the function $%
\psi \left( \mathbf{x}\right) $ over the interval $\left(
s_{n},s_{n-1}\right) $. Then we approximate the boundary condition (\ref{3.6}%
) as a piecewise constant function with respect to $s$,
\begin{equation}
q_{n}\left( \mathbf{x}\right) =\psi _{n}\left( \mathbf{x}\right) ,\mathbf{x}%
\in \partial \Omega .  \label{3.9}
\end{equation}%
Using (\ref{3.8}), integrate equation (\ref{3.5}) with respect to $s\in %
\left[ s_{n},s_{n-1}\right) $. We obtain for $n\geq 1$
\begin{equation}
\begin{array}{c}
\Delta q_{n}+A_{2,n}\left( h\sum\limits_{j=0}^{n-1}\triangledown
q_{j}-\triangledown T\right) \triangledown q_{n}-A_{1,n}\left(
\triangledown
q_{n}\right) ^{2}= \\
A_{3,n}h\sum\limits_{j=1}^{n-1}\Delta q_{j}+A_{4,n}\left(
h\sum\limits_{j=0}^{n-1}\triangledown q_{j}-\triangledown T\right)
^{2}-A_{3,n}\Delta T,%
\end{array}
\label{3.10}
\end{equation}%
where $A_{k,n},k=1,...,4$ are certain coefficients, see \cite{bib18}
for exact formulas for them. Let
\begin{equation}
\overline{s}>2,h\in \left( 0,1\right) .  \label{3.101}
\end{equation}%
Then one can prove that $\left\vert A_{i,n}\right\vert \leq 8\overline{s}%
^{2},i=2,3,4,$
\begin{equation}
\left\vert A_{1,n}\right\vert \leq 2\overline{s}^{2}h.
\label{3.102}
\end{equation}%
Assuming that $h$ is sufficiently small and using (\ref{3.102}), we
assume below that
\begin{equation}
A_{1,n}\left( \triangledown q_{n}\right) ^{2}:=0\text{ in
(\ref{3.10}).} \label{3.103}
\end{equation}

It is convenient for our convergence analysis to formulate the
Dirichlet
boundary value problem (\ref{3.9}), (\ref{3.10}) in the weak form. Denote $%
H_{0}^{1}(\Omega) =\{ u\in H^{1}(\Omega) :u\mid _{\partial \Omega
}=0\}$. Assume that there exists functions $\Psi _{n}$ such that
\begin{equation}
\Psi _{n}\in H^{2}\left( \Omega \right) \text{ and }\Psi _{n}\mathbf{\mid }%
_{\partial \Omega }=\psi _{n}\left( \mathbf{x}\right) ,n\in \left[
1,N\right] \text{.}  \label{3.11}
\end{equation}%
Consider the function $p_{n}=q_{n}-\Psi _{n}\in H_{0}^{1}\left(
\Omega \right) .$ Then we obtain an obvious analog of equation
(\ref{3.10}) for the function $p_{n}.$ Multiply both sides of the
latter equation by an arbitrary
function $\eta \in H_{0}^{1}\left( \Omega \right) $ and integrate over $%
\Omega $ using integration by parts and (\ref{3.103}). We obtain
\begin{eqnarray}
&&\int\limits_{\Omega }\nabla p_{n}\nabla \eta d\mathbf{x}%
-A_{2,n}\int\limits_{\Omega }\left(
h\sum\limits_{j=0}^{n-1}\triangledown q_{j}-\triangledown T\right)
\triangledown p_{n}\cdot \eta d\mathbf{x}
\notag \\
&=&-\int\limits_{\Omega }\left\{ \nabla \Psi _{n}\nabla \eta +\left[
A_{2,n}\left( h\sum\limits_{j=0}^{n-1}\triangledown
q_{j}-\triangledown T\right) \triangledown \Psi _{n}+f_{n}\right]
\eta \right\} d\mathbf{x} \label{3.22}
\end{eqnarray}%
\begin{equation*}
\mathbf{+}A_{3,n}\int\limits_{\Omega }h\sum\limits_{j=0}^{n-1}\nabla
q_{j}\nabla \eta d\mathbf{x,}\text{ }p_{n}\in H_{0}^{1}\left( \Omega
\right) ,\forall \eta \in H_{0}^{1}\left( \Omega \right) ,
\end{equation*}%
\begin{equation*}
f_{n}\left( \mathbf{x}\right) =A_{4,n}\left(
h\sum\limits_{j=0}^{n-1}\triangledown q_{j}-\triangledown T\right)
^{2}-A_{3,n}\Delta T.
\end{equation*}%
Hence, the function $q_{n}\in H^{1}\left( \Omega \right) $ is a weak
solution of the problem (\ref{3.9}), (\ref{3.10}) if and only if the
function $p_{n}\in H_{0}^{1}\left( \Omega \right) $ satisfies the
integral identity (\ref{3.22}). The question about existence and
uniqueness of the
weak solution of the problem (\ref{3.22}) is addressed in Theorem \ref%
{th:4.1}.

We now describe our algorithm of sequential solutions of boundary
value problems (\ref{3.9}), (\ref{3.10}) for $n=1,\ldots ,N$,
assuming that an approximation $T\left( \mathbf{x}\right) $ for the
tail function is found (see subsection 5.2 for the latter). Recall
that $q_{0}=0.$ Hence, we have:

\emph{Step }$n\in \left[ 1,N\right] $. Suppose that functions
$q_{1},\ldots ,q_{n-1}$ are computed. On this step we find the weak
solution of the Dirichlet boundary problem (\ref{3.9}), (\ref{3.10})
for the function $q_{n}$ via the FEM with triangular finite
elements.

\emph{Step} $N+1$. After functions $q_{1},\ldots ,q_{N}$ are
computed, the function $a\left( \mathbf{x}\right) $ is reconstructed
using backwards calculations at $n:=N$, i.e., for the lowest value
of $s:=s_{N}=\underline{s} $ as described in subsection 5.1.

\section{Convergence Analysis}

\subsection{Approximate global convergence}

The central question which was addressed in above cited publications \cite%
{bib18,bib22,bib23,bib24} was of constructing such a numerical
method for the CIP (\ref{2.1})-(\ref{2.4}), which would
simultaneously satisfy the following three conditions:

1. This method should deliver a good approximation for the exact
solution of this CIP without an a priori knowledge of a small
neighborhood of this solution.

2.\emph{\ }A theorem should be proven which would guarantee of
obtaining such an approximation.

3. A good numerical performance of this technique should be
demonstrated on computationally simulated data and, optionally, on
real data.

We use the word \textquotedblleft optionally", because it is usually
hard and expensive to actually collect real data. It is
\emph{enormously challenging} to construct such a numerical method.
Therefore, as it is often done in mathematical modeling,
\emph{conceptually}, our approach consists of the following six
steps:

\textbf{Step 1.} A reasonable approximate mathematical model is
proposed. The accuracy of this model cannot be rigorously estimated.

\textbf{Step 2.} A numerical method is developed, which works within
the framework of this model.

\textbf{Step 3.} A theorem is proven, which guarantees that, within
the framework of this model, the numerical method of Step 2 indeed
reaches a sufficiently small neighborhood of the exact solution.
Naturally, the smallness of this neighborhood depends on the level
of the error, both in the data and in some additional approximations
is sufficiently small.

\textbf{Step 4.} The numerical method of Step 2 is tested on
computationally simulated data.

\textbf{Step 5}. The numerical method of Step 2 is tested on real
data.

\textbf{Step 6.} Finally, if results of Steps 4 and 5 are good ones,
then that approximate mathematical model is proclaimed as a valid
one.

Thus, results of the current paper (Step 5) in combination with the
previous ones of \cite{bib18,bib22,bib23,bib24} lead to the positive
conclusion of Step 6. Step 6 is logical, because its condition is
that the resulting numerical method is proved to be effective. It is
sufficient to achieve that small neighborhood of the exact solution
after a finite (rather than infinite) number of iterations. We refer
to page 157 of the book \cite{bib14} where it is stated that the
number of iterations can be regarded as a regularization parameter
sometimes for an ill-posed problem. Still, in our computations both
for simulated and real data, classical convergence in the Cauchy
sense was achieved. These consideration lead to the following
definition of the approximate global convergence property.

\begin{definition}
\label{de:4.1} (approximate global convergence)
\cite{bibBook,bib20}. Consider a\textbf{\ }nonlinear ill-posed
problem $P$. Suppose that this
problem has a unique solution $x^{\ast }\in B$ for the noiseless data $%
y^{\ast },$ where \ $B$ is a Banach space with the norm $\left\Vert
\cdot \right\Vert _{B}.$ We call $x^{\ast }$ \textquotedblleft exact
solution\textquotedblright\ or \textquotedblleft correct
solution\textquotedblright . Suppose that a certain approximate
mathematical model $\overline{M}$ is proposed to solve the problem
$P$ numerically. Assume that, within the framework of the model
$\overline{M},$ this problem has unique exact solution
$x_{\overline{M}}^{\ast }.$ Also, let one of assumptions of the
model $\overline{M}$ be that $x_{M_{1}}^{\ast }=x^{\ast }. $
Consider an iterative numerical method for solving the problem $P$.
Suppose that this method produces a sequence of points $\left\{
x_{n}\right\} _{n=1}^{N}\subset B,$ where $N\in \left[ 1,\infty
\right) .$ Let the number $\theta \in \left( 0,1\right) .$ We call
this numerical
method \emph{approximately globally convergent of the level }$\theta $\emph{%
, }or shortly\emph{\ globally convergent}, if, within the framework
of the approximate model $\overline{M},$ a theorem is proven, which
guarantees that, without any \emph{a priori} knowledge of a
sufficiently small neighborhood of $x^{\ast },$ there exists a
number $\overline{N}\in \left[
1,N\right) $ such that%
\begin{equation}
\left\Vert x_{n}-x^{\ast }\right\Vert _{B}\leq \theta ,\forall n\in
\left[ \overline{N},N\right] .  \label{4.1}
\end{equation}%
Suppose that iterations are stopped at a certain number $k\in \left[
\overline{N},N\right] .$ Then the point $x_{k}$ is denoted as $%
x_{k}:=x_{glob}$ and is called \textquotedblleft the approximate
solution resulting from this method".
\end{definition}

With reference to the notion of the approximate mathematical model $%
\overline{M}$, as well as to the above Steps 1-6, it is worthy to
mention
here that one of the keys to the successful numerical implementation \cite%
{bib2} of the non-local reconstruction algorithm of \cite{bib30} was
the use of a certain approximate mathematical model. The same is
true for the two dimensional analog of the Gel'fand-Levitan-Krein
equation being applied in \cite{bib27} to an inverse hyperbolic
problem. Thus, it seems that whenever one is trying to construct an
efficient non-locally convergent numerical method for a truly
complicated nonlinear ill-posed problem, it is close to the
necessity to introduce some approximations which cannot be
rigorously justified and then work within the resulting approximate
mathematical model then. Analogously, although the Huygens-Fresnel
theory of optics is not rigorously supported by the Maxwell's system
(see section 8.1 of the classical textbook \cite{BW}), the
\textquotedblleft diffraction part" of the entire modern optical
industry is based on the Huygens-Fresnel optics.

\subsection{ Exact solution}

By one of statements of Introduction as well as by Definition 4.1,
we should assume the existence and uniqueness of an
\textquotedblleft ideal\textquotedblright\ exact solution $a^{\ast
}\left( \mathbf{x}\right) $ of our inverse problem for an
\textquotedblleft ideal\textquotedblright\
noiseless exact data $\varphi ^{\ast }\left( \mathbf{x,x}_{0}\right) $ in (%
\ref{2.4}).\ Next, in accordance with the regularization theory, one
should assume the presence of an error in the data of a small level
$\gamma >0$ and construct an approximate solution for this case.

Since the exact solution was defined in \cite{bib18}, we outline it
only briefly here for the convenience of the reader. Let the
function $a^{\ast }\left( \mathbf{x}\right) $ satisfying conditions
(\ref{2.3}) be the exact solution of our inverse problem for the
noiseless data $\varphi ^{\ast
}\left( \mathbf{x},\mathbf{x}_{0}\right) $ in (\ref{2.4}). We assume that $%
a^{\ast }\left( \mathbf{x}\right) $ is unique. Let the function
$u^{\ast
}\left( \mathbf{x},s\right) $ be the same as the function $u\left( \mathbf{x}%
,\mathbf{x}_{0}\right) $ in Theorem \ref{th:2.1}, but for the case
$a\left( \mathbf{x}\right) :=a^{\ast }\left( \mathbf{x}\right) $.
Denote
\begin{eqnarray*}
w^{\ast }\left( \mathbf{x},s\right)  &=&s^{-2}\ln u^{\ast }\left( \mathbf{x}%
,s\right) ,\text{ }q^{\ast }\left( \mathbf{x},s\right) =\partial
_{s}w^{\ast
}\left( \mathbf{x},s\right) , \\
T^{\ast }\left( \mathbf{x}\right)  &=&w^{\ast }\left( \mathbf{x},\overline{s}%
\right) .
\end{eqnarray*}%
The function $q^{\ast }$ satisfies the analogue of equation
(\ref{3.5}) with the boundary condition (\ref{3.6})
\begin{equation}
\begin{array}{c}
q^{\ast }\left( \mathbf{x},s\right) =\psi ^{\ast }\left(
\mathbf{x},s\right) ,\forall \left( \mathbf{x},s\right) \in \partial
\Omega \times \left[
\underline{s},\bar{s}\right] ,%
\end{array}
\label{4.2}
\end{equation}%
where by (\ref{3.7}) $\psi ^{\ast }\left( \mathbf{x},s\right)
=\partial
_{s}\left( s^{-2}\ln u^{\ast }\left( \mathbf{x},s\right) \right) $ for $%
\left( \mathbf{x},s\right) \in \partial \Omega \times \left[ \underline{s},%
\bar{s}\right] $. We call $q^{\ast }\left( \mathbf{x},s\right) $
\emph{the exact solution.} It should be noted that our real data
$\varphi \left( \mathbf{x,}s\right) $ are naturally given with a
random noise. It is well known that the differentiation of the noisy
data is an ill-posed problem \cite{bibBook,bib25}. However, because
of some limitations of our device, we work only with three values of
the parameter $s$ in our real data (subsection 5.2). Hence, we
simply calculate the $s-$ derivative via the finite difference.\ Our
numerical experience shows that this does not lead to a degradation
of our reconstruction results.

\subsection{Our approximate mathematical model}

In this model we basically assume that the upper bound for the tail
function can become sufficiently small independently on
$\overline{s}.$ In addition, since we calculate the tail function
$T\left( \mathbf{x},\overline{s}\right) $ separately from the
function $q$, we assume that the function $T\left(
\mathbf{x},\overline{s}\right) $ is given (but not the exact tail function $%
T^{\ast }\left( \mathbf{x},\overline{s}\right) ).$ Although the
smallness of $T$ is supported by (\ref{3.4}), the independence of
that upper bound of this function from $\overline{s}$ does not
follow from (\ref{3.4}). Still this is one of two assumptions of our
\emph{approximate} mathematical model. Following the concept of
Steps 1-6 of subsection 4.1, we now verify this model on real data.

More precisely, our approximate mathematical model $\overline{M}$
consists of the following two assumptions

\textbf{Assumptions:}

\textbf{1. }We assume that the number $\overline{s}>1$ is
sufficiently large and fixed. Also, the tail function $T\left(
\mathbf{x},\overline{s}\right) $ is given, and tail functions
$T^{\ast }\left( \mathbf{x},\overline{s}\right) ,T\left(
\mathbf{x},\overline{s}\right) $ have the forms
\begin{eqnarray}
T^{\ast }\left( \mathbf{x},\overline{s}\right)  &=&\frac{\ln
w_{0}\left(
\overline{s}\right) }{\overline{s}^{2}}+r^{\ast }\left( \mathbf{x},\overline{%
s}\right) ,\text{ }\mathbf{x}\in \Omega ,\text{ }r^{\ast }\in
C^{2}\left(
\overline{\Omega }\right) ,  \label{4.40} \\
T\left( \mathbf{x},\overline{s}\right)  &=&\frac{\ln w_{0}\left( \overline{s}%
\right) }{\overline{s}^{2}}+r\left( \mathbf{x},\overline{s}\right) ,\text{ }%
\mathbf{x}\in \Omega ,\text{ }r\in C^{2}\left( \overline{\Omega
}\right) ,
\label{4.41} \\
\left\Vert r^{\ast }\right\Vert _{C^{2}\left( \overline{\Omega
}\right) }
&\leq &\xi ,\text{ }\left\Vert r\right\Vert _{C^{2}\left( \overline{\Omega }%
\right) }\leq \xi ,  \label{4.42}
\end{eqnarray}%
where $\xi \in \left( 0,1\right) $ is a small parameter.
Furthermore, the parameter $\xi $ is independent on $\overline{s}.$

\textbf{2}. As the unknown vector $x\in B$ in Definition 4.1, we
choose the vector function $\overline{q}\left( \mathbf{x}\right)
=\left( q_{1},...,q_{N}\right) \left( x\right) $ in (\ref{3.8}),
(\ref{3.10}). In
this case we will have only one iteration in Definition 4.1, i.e. $N=%
\overline{N}=1$ in (\ref{4.1}).

It is for the sake of our convergence analysis that we choose here
the vector function $\overline{q}\left( \mathbf{x}\right) $ rather
than the coefficient $a\left( \mathbf{x}\right) $ as our unknown
function. Indeed,
finding $a\left( \mathbf{x}\right) $ in the weak form (\ref{5.1}), (\ref{5.2}%
) would require some interpolation estimates of the differences
between functions $p_{n}^{\ast }$ and their interpolations via the
finite dimensional subspace $G_{m}$ (see subsection 4.4 for
$p_{n}^{\ast },G_{m}$). The latter has an \textquotedblleft
underwater rock" in terms of the convergence analysis.

By Definition 4.1 we need to prove uniqueness of the function
$q^{\ast }\left( \mathbf{x},s\right) $ for $\mathbf{x}\in \Omega
,s\in \left[
\underline{s},\overline{s}\right] .$ Assuming that the exact tail function $%
T^{\ast }\left( \mathbf{x},\overline{s}\right) $\ satisfying (\ref{4.40}), (%
\ref{4.42}) is given and the interval $\left[ \underline{s},\overline{s}%
\right] $ is sufficiently small, this can be done easily, using the
fact that $q^{\ast }\left( \mathbf{x},s\right) $ is the solution of
an obvious analog of the problem (\ref{3.5}), (\ref{3.6}). This is
because Volterra-like integrals are involved in (\ref{3.5}).

\subsection{Approximate global convergence theorem}

Let $\gamma $ be the sum of the following errors: the level of the
error in the function $\psi (\mathbf{x},s)$ (by (\ref{3.7}) this
function is generated by the measurement data), errors in our finite
element approximations of $q_{n}$, the magnitude of the
$C^{2}(\overline{\Omega })$
norm of the tail function as well as the length of the interval $(\underline{%
s},\overline{s})$.

Consider a finite dimensional subspace $G_{m}\subset H_{0}^{1}\left(
\Omega \right) $ with $\dim G_{m}=m$. A realistic example of $G_{m}$
is the subspace, generated by triangular finite elements, which are
piecewise linear functions. We assume that $f_{x},f_{z}\in L_{\infty
}\left( \Omega \right) ,\forall f\in G_{m}.$ Since all norms in the
subspace $G_{m}$ are equivalent, then there exists a constant
$C_{m}=C_{m}\left( G_{m}\right) $ such that
\begin{equation}
\left\Vert \nabla f\right\Vert _{L_{\infty }\left( \Omega \right)
}\leq
C_{m}\left\Vert \nabla f\right\Vert _{L_{2}\left( \Omega \right) },\text{ }%
\forall f\in G_{m},\text{ }C_{m}\geq 2.  \label{4.11}
\end{equation}%
Here and below $\left\Vert \nabla f\right\Vert _{L_{\infty }\left(
\Omega \right) }:=\left( \left\Vert f_{x}\right\Vert _{L_{\infty
}\left( \Omega \right) }^{2}+\left\Vert f_{y}\right\Vert _{L_{\infty
}\left( \Omega \right) }^{2}\right) ^{1/2},$ and the same for the
$L_{2}\left( \Omega \right) $ norm. Theorem \ref{th:4.1} works with
the weak formulation (\ref{3.22}). In
the course of the proof of this theorem we need to estimate products $%
\triangledown q_{j}\triangledown p_{n}.$While this was done for
$C^{2+\alpha }\left( \overline{\Omega }\right) $ solutions in
\cite{bib18} using Schauder theorem, in the weak formulation we need
to assume that functions
\begin{equation}
q_{n}-\Psi _{n}:=p_{n}\in G_{m}.  \label{4.12}
\end{equation}%
We introduce functions $q_{n}^{\ast },\Psi _{n}^{\ast },$ which are
analogs of functions $q_{n},\Psi _{n}$ for the case $a:=a^{\ast }.$
Let $p_{n}^{\ast }:=$ $q_{n}^{\ast }-\Psi _{n}^{\ast }.$ We assume
that
\begin{equation}
q_{n}^{\ast },\Psi _{n}^{\ast }\in C^{2}\left( \overline{\Omega
}\right)
,\left\Vert \nabla q_{n}^{\ast }\right\Vert _{C\left( \overline{\Omega }%
\right) }\leq C^{\ast },n\in \left[ 1,N\right] ,  \label{4.13}
\end{equation}%
\begin{equation}
\max_{s\in \left[ s_{n},s_{n-1}\right] }\left( \left\Vert q^{\ast
}\left( \mathbf{x},s\right) -q_{n}^{\ast }\left( \mathbf{x}\right)
\right\Vert _{C^{2}\left( \overline{\Omega }\right) }\right) \leq
C^{\ast }h,n\in \left[ 1,N\right] ,  \label{4.14}
\end{equation}%
\begin{equation}
A_{1,n}\left( \triangledown q_{n}^{\ast }\right) ^{2}:=0,
\label{4.15}
\end{equation}%
\begin{equation}
p_{n}^{\ast }\in G_{m},\left\Vert \nabla p_{n}^{\ast }\right\Vert
_{C\left( \overline{\Omega }\right) }\leq 2C^{\ast },n\in \left[
1,N\right] , \label{4.16}
\end{equation}%
\begin{equation}
\left\Vert \nabla \Psi _{n}\left( \mathbf{x}\right) -\nabla \Psi
_{n}^{\ast }\left( \mathbf{x}\right) \right\Vert _{L_{2}\left(
\Omega \right) }\leq C^{\ast }\left( \sigma +h\right) ,n\in \left[
1,N\right] ,  \label{4.17}
\end{equation}%
where the number $C^{\ast }>1$ is given and $\sigma >0$ is a small
parameter
characterizing the level of the error in the data $\psi \left( \mathbf{x}%
,s\right)$. The connection between $\sigma $ and the error in $\psi
\left(\mathbf{x},s\right)$ is clear from
comparison of (\ref{3.9}), (\ref{3.11}) with (\ref{4.17}). Condition (\ref%
{4.15}) is a direct analog of (\ref{3.103}). Each function
$p_{n}^{\ast }$ is the weak solution of an analog of the problem
(\ref{3.22}).

\begin{theorem}
\label{th:4.1}\emph{Let }$\Omega \subset \mathbb{R}^{2}$\emph{\ be a
convex bounded domain with the boundary }$\partial \Omega \in
C^{3}$\emph{. Assume that Assumptions 1,2 of subsection 4.3 hold. In
addition, assume that conditions (\ref{3.101}), (\ref{3.103}),
(\ref{4.12})-(\ref{4.17}) hold. Let
}$\beta =\overline{s}-\underline{s}$. \emph{Consider the error parameter }$%
\gamma =\beta +h+\sigma +\xi ,$ where $\xi $ is the small parameter
defined in (\ref{4.40})-(\ref{4.42}), which is independent on
$\overline{s}$
(Assumption 1). \emph{Let }$C_{m}\geq 1$\emph{\ be the constant defined in (%
\ref{4.11}). Then }there exists a constant $B=B\left( \Omega \right) >0$%
\emph{\ such that if }
\begin{equation}
\gamma \leq \frac{B}{\overline{s}^{2}C^{\ast }},  \label{4.18}
\end{equation}%
then\emph{\ for each }$n\in \left[ 1,N\right] $\emph{\ there exists
unique solution }$p_{n}\in G_{m}$\emph{\ of the problem (\ref{3.22})
and for
functions }$q_{n}=p_{n}+\Psi _{n}$ \emph{the following estimate holds }%
\begin{equation}
\left\Vert \triangledown q_{n}-\triangledown q_{n}^{\ast
}\right\Vert _{L_{2}\left( \Omega \right) }\leq \theta ,
\label{4.19}
\end{equation}%
where $\theta =\left( BC^{\ast }\overline{s}^{2}\right) \gamma \in
\left( 0,1\right) .$ Thus, \emph{(\ref{4.19}) implies that }our
numerical method has the approximate global convergence property of
the level $\theta $ (Definition 4.1).
\end{theorem}

\textbf{Remarks 4.1.}

1. We omit the proof of this theorem, because it is similar with the
proof of a similar theorem in \cite{bib18} for the $C^{\alpha }$
space. The only technical challenge in theorem \ref{th:4.1} is to
obtain analogous arguments in the finite dimensional space $G_{m}.$

2. The smallness of the parameter $\beta $ in (\ref{4.18}) is a
natural requirement since the original equation (\ref{3.5}) contains
Volterra-like integrals in nonlinear terms. It is well known from
the standard ODE course that the existence of a Volterra-like
nonlinear integral equation of the second kind can be proven only on
a small interval.

\section{Numerical Implementation}

Since this implementation was described in detail in \cite{bib18},
we outline it only briefly here for the convenience of the reader.
As to the functions $q_{n},$ we have sequentially calculated them
via the FEM solving Dirichlet boundary value problems (\ref{3.9}),
(\ref{3.10}). As it was mentioned in the end of subsection 3.2, we
have used standard triangular finite elements. Two important
questions which are discussed in this section are about
approximating the function $a\left( \mathbf{x}\right) $ and the tail
function $T\left( \mathbf{x}\right) $.

\subsection{Approximation of the function $a\left( \mathbf{x}\right) $}

We reconstruct the target coefficient $a\left( \mathbf{x}\right) $
via
backwards calculations as follows. First, we reconstruct the function $%
w\left( x,s_{N}\right) $ from $q_n,n=1,2,\ldots,N$ and tail $T$. We have $%
u\left( \mathbf{x},s_{N}\right) =\exp \left[ s^2w\left( \mathbf{x}%
,s_{N}\right) \right] $. Next, since in (\ref{2.1}) the source $\mathbf{x}%
_{0}\notin \overline{\Omega },$ we use equation (\ref{2.1}) in the
weak form as
\begin{equation}
-\int\limits_{\Omega }\nabla u\triangledown \eta_{p}d\mathbf{x}%
=\int\limits_{\Omega }au\eta_{p}d\mathbf{x},  \label{5.1}
\end{equation}
where the test function $\eta_{p}\left( \mathbf{x}\right) ,p\in \left[ 1,P%
\right] $ is a quadratic finite element of a computational mesh with
the
boundary condition $\eta_{p}\left( \mathbf{x}\right) |_{\partial \Omega }=0$%
. The number $P$ is finite and depends on the mesh we choose. Equalities (%
\ref{5.1}) lead to a linear algebraic system which we solve. Since
this formulation is complex but standard, it is omitted here.
Interested readers can see \cite{bib13}. Finally, we let
\begin{equation}
a\left( \mathbf{x}\right) = \max \left( \overline{a}\left( \mathbf{x}%
\right), k^{2}\right) .  \label{5.2}
\end{equation}

\subsection{Construction of the tail function}

The above construction of functions $q_{n}$ depends on the tail
function $T.$ In this subsection we state briefly our procedure of
approximating the tail function, see \cite{bib18,bib22,bib23} for
details. This procedure consists of two stages. First, we find a
first guess for the tail using the asymptotic behavior of the
solution of the problem (\ref{2.1}), (\ref{2.2})
as $\left\vert \mathbf{x}_{0}\right\vert \rightarrow \infty $ (Lemma \ref%
{Le:2.2}), as well as boundary measurements. On the second stage we
refine the tail.

\begin{figure}[th]
\begin{center}
\begin{minipage}{130mm}
\subfigure[]{ \resizebox*{6cm}{!}{\label{F1a}\includegraphics{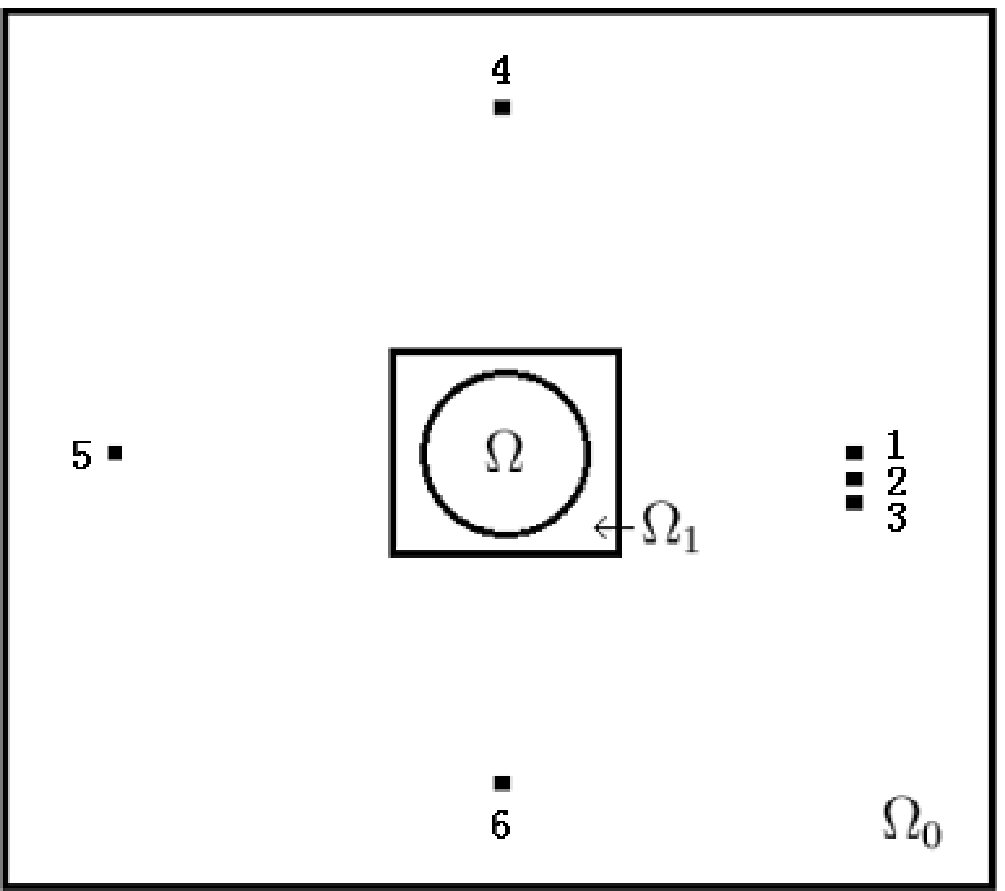}}} %
\subfigure[]{ \resizebox*{6cm}{!}{\label{F1b}\includegraphics{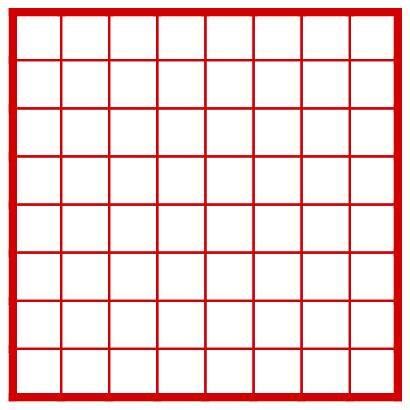}}} %
\caption{\emph{The schematic diagrams of inverse problem domain and
light source locations. (a) Illustrates a layout of the inverse
problem setting in 2-D. The circular disk $\Omega $ corresponds to a
horizontal cross-section of the hemisphere part of phantom (the
supposed \textquotedblleft mouse head" in animal experiments), The
computational domain $\Omega _{1}$ is a rectangle containing $\Omega
$ inside, 6 light sources are located outside of the computational
domain. Because of limitations of our device, we use only three
locations of the light source along one line (number 1,2,3) to model
the source }$\mathbf{x}_{0}$\emph{\ running along the straight line
}$L $. \emph{Light sources numbered 1,4,5,and 6 are used to
construct an
approximation for the tail function. (b) Depicts the computational domain }$%
\Omega _{1}=\{(x,z),|x|<5.83,|z|<5.83\}$ \emph{(unit:mm) and its
rectangular meshes for tail functions and inverse calculations for
the numerical method of this paper. Actual mesh is much dense than
these displayed. The diagram is not scaled to actual sizes.}}
\label{F1}
\end{minipage}
\end{center}
\end{figure}

We display in Fig. \ref{F1} the boundary data collection scheme in
our
experiment. We have used six locations of the light sources, see Fig. \ref%
{F1a}. Sources number 1,2 and 3 are the ones which model the source $\mathbf{%
x}_{0}$ running along the straight line $L,$ see (\ref{2.4}). The
distance between these sources was six (6) millimeters. Each light
source means drilling a small hole in the phantom to fix the source
position. It was impossible to place more sources in a phantom of
this size that mimics actual mouse head. So, since the data for the
functions $q_{n}$ are obtained via the differentiation with respect
to the source position, we have used only two functions
$q_{1},q_{2}$. Fortunately, the latter was
sufficient for our goal of imaging of abnormalities. Sources 1,4,5,6 (Fig. %
\ref{F1a}) were used to approximate the tail function as described
below. We
construct that approximation in the domain $\Omega _{1}$ depicted on Fig. %
\ref{F1b}. This domain is (units in millimeters)
\begin{equation}
\Omega _{1}:=\left\{ \mathbf{x}=\left( x,z\right)
:x_{1}=-5.83<x<x_{2}=5.83,z_{1}=-5.83<z<z_{2}=5.83\right\} .
\label{5.3}
\end{equation}

\subsubsection{The first stage of approximating the tail}

Since this stage essentially relies on positions of sources 1,4,5,6,
it makes sense to describe this stage in detail here, although it
was also described in \cite{bib18}. The light sources are placed as
far as possible
from the domain $\Omega _{1}$, so that the asymptotic approximations (\ref%
{5.5}), (\ref{5.6}) of solutions are applicable. Let
$\mathbf{s}^{\left( j\right) }$ be the two-dimensional vector
characterizing the position of the light source number $j$. First,
consider the light source number 1. Denote
\begin{equation}
S^{\left( 1\right) }:=S^{\left( 1\right) }\left(
x,z,\mathbf{s}^{\left( 1\right) }\right) =\left\vert \left(
x,z\right) -\mathbf{s}^{\left( 1\right) }\right\vert .  \label{5.4}
\end{equation}%
Using (\ref{2.5}), Theorem \ref{th:2.1} and Lemma \ref{Le:2.2}, we
obtain
for the function $\widetilde{w}=\ln u$%
\begin{equation}
\widetilde{w}\left( x,z,\mathbf{s}^{\left( 1\right) }\right)
=-kS^{\left( 1\right) }-\ln \left( 2\sqrt{2\pi }\right)
-\frac{1}{2}\ln S^{\left( 1\right) }+p_{\infty }\left( x,z\right)
+O\left( \frac{1}{S^{\left( 1\right) }}\right) ,S^{\left( 1\right)
}\rightarrow \infty .  \label{5.5}
\end{equation}%
Since by (\ref{5.3}) $\widetilde{w}\left( x,z_{1},\mathbf{s}^{\left(
1\right) }\right) =\varphi \left( \mathbf{x,s}^{\left( 1\right) }\right) ,%
\mathbf{x}\in \partial \Omega \cap \left\{ z=z_{1}\right\} ,$ we use (\ref%
{5.5}) to approximate the unknown function $p_{\infty }\left(
x,z_{1}\right) $ as
\begin{equation}
p_{\infty }\left( x,z_{1}\right) =\widetilde{w}\left( x,z_{1},\mathbf{s}%
^{(1)}\right) +kS\left( x,z_{1},\mathbf{s}^{(1)}\right)
+\frac{1}{2}\ln \left( \frac{\pi }{2S\left(
x,z_{1},\mathbf{s}^{(1)}\right) }\right) . \label{5.6}
\end{equation}

Formula (\ref{5.6}) gives the value of $p_{\infty }\left(
x,z_{1}\right) $ only at $z=z_{1}$. Since $\Omega _{1}$ is a square,
we set the first guess for the tail as the one which is obtained
from (\ref{5.6}) by simply
extending the values at $z=z_{1}$ to the entire domain of $\Omega _{1},$%
\begin{equation*}
\widetilde{w}\left( x,z,\mathbf{s}^{\left( 1\right) }\right)
=-kS^{\left(
1\right) }\left( x,z,\mathbf{s}^{\left( 1\right) }\right) -\ln \left( 2\sqrt{%
2\pi }\right) -\frac{1}{2}\ln S^{\left( 1\right) }\left( x,z,\mathbf{s}%
^{\left( 1\right) }\right) +p_{\infty }\left( x,z_{1}\right) ,
\end{equation*}
where the function $S\left( x,z,\mathbf{s}^{\left( 1\right) }\right)
$ is
given in (\ref{5.4}). Next, we compute the function $u\left( x,z,\mathbf{s}%
^{\left( 1\right) }\right) =\exp \left( \widetilde{w}\left( x,z,\mathbf{s}%
^{\left( 1\right) }\right) \right) $ and get $a^{\left( 1\right)
}\left( x,z\right) ,\left( x,z\right) \in \Omega $ via (\ref{5.1}),
(\ref{5.2}).

For light sources 4-6, we repeat the above procedure to get
$a^{\left(
4\right) }\left( x,z\right) ,a^{\left( 5\right) }\left( x,z\right) $ and $%
a^{\left( 6\right) }\left( x,z\right) $ respectively. Then we
consider the average coefficient $\overline{a}(x,z)$ and set (see
(\ref{5.2}))
\begin{equation}
a\left( x,z\right) :=\max (\overline{a}(x,z),k^{2}).  \label{5.8}
\end{equation}%
Next, we solve the forward problem (\ref{2.1}), (\ref{2.2}) for the light $%
\mathbf{s}^{\left( 3\right) }$ with this coefficient $a\left(
x,z\right) $ again to get $u\left( x,z,\overline{s}\right)
,\overline{s}=\left\vert \mathbf{s}^{\left( 3\right) }\right\vert .$
The final approximate tail function obtained on the first stage is
\begin{equation}
T_{1}\left( x,z\right) =\frac{\ln u\left( x,z,\overline{s}\right) }{%
\overline{s}^{2}}.  \label{5.9}
\end{equation}

\subsubsection{The second stage for the tail}

The second stage involves an iterative process that enhances the
first approximation for the tail (\ref{5.9}). We describe this stage
only briefly here referring to \cite{bib18} for details. In this
case we use only one source number 3, $\mathbf{s}^{\left( 3\right)
}.$\ Recall that $\left\vert
\mathbf{s}^{\left( 3\right) }\right\vert =\overline{s}$. Let the function $%
a\left( x,z\right) $ be the one calculated in (\ref{5.8}). Denote $%
a_{1}\left( \mathbf{x}\right) :=a\left( x,z\right) .$ Next, we solve
the following boundary value problem
\begin{eqnarray*}
\Delta u_{1}-a_{1}\left( \mathbf{x}\right) u_{1} &=&0,\mathbf{x}\in
\Omega
_{1}, \\
u_{1}|_{\partial \Omega _{1}} &=&\varphi \left( \mathbf{x},\overline{s}%
\right) ,\mathbf{x}\in \partial \Omega _{1}.
\end{eqnarray*}%
Then, we iteratively solve the following boundary value problems
\begin{eqnarray*}
\Delta w_{m}-a_{m}\left( \mathbf{x}\right) w_{m} &=&\left[
a_{m}\left(
\mathbf{x}\right) -a_{m-1}\left( \mathbf{x}\right) \right] u_{m-1},\text{ }%
m\geq 2, \\
w_{m}|_{\partial \Omega _{1}} &=&0.
\end{eqnarray*}%
We set $u_{m}:=u_{m-1}+w_{m}.$ Next, using (\ref{5.1}) and (\ref{5.2}) with $%
u:=u_{m},$ we find the function $a_{m+1}\left( \mathbf{x}\right) .$
We have computationally observed that this iterative process
provides a convergent sequence $\left\{ a_{m}\left(
\mathbf{x}\right) \right\} $ in $L_{2}\left( \Omega _{1}\right) .$
We stop iterations at $m:=m_{1}$, where $m_{1}$ is defined via
\begin{equation}
\frac{\left\Vert a_{m_{1}}-a_{m_{1}-1}\right\Vert _{L_{2}\left(
\Omega _{1}\right) }}{\left\Vert a_{m_{1}-1}\right\Vert
_{L_{2}\left( \Omega _{1}\right) }}\leq \varepsilon ,  \label{5.10}
\end{equation}%
where $\varepsilon >0$ is a small number of our choice, and the norm in $%
L_{2}\left( \Omega _{1}\right) $ is understood in the discrete
sense. Next, assuming that $u_{m_{1}}>0$, we set for the tail
\begin{equation}
T\left( \mathbf{x}\right) =\frac{1}{\overline{s}^{2}}\ln
u_{m_{1}}\left( \mathbf{x}\right) .  \label{5.11}
\end{equation}%
Then, using (\ref{5.11}), we proceed with calculating of functions
$q_{n}$ as described above.

\section{Real Data}

We now describe our experimental setup for collecting the optical
tomography data from an optical phantom.\ This phantom is a man-made
subject that has the same optical property as small animals. Such a
phantom is a well-accepted standard to test reconstruction methods
for real applications before animal experiments.

\begin{figure}[th]
\centering
\includegraphics[width=0.4\textwidth]{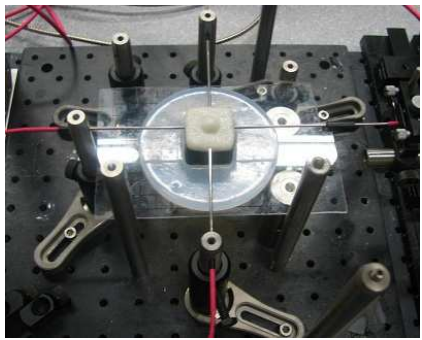}
\caption{\emph{A photograph of the experimental setup. An optical
phantom is connected with 4 laser fibers, the one on the right hand
side is movable to locations 1,2, and 3 (Fig. \ref{F1a}). The
gelatin made phantom has the shape of the rectangular box with a
hemisphere on its top surface. The hemisphere mimics the head of a
mouse with a mask exposing the crown part of the mouse. A CCD camera
mounted above the phantom (not shown) provides light intensity
measurements of the top surface of the phantom.}} \label{F2}
\end{figure}

Fig. \ref{F2} is a photograph of our measurement setup. The center
of the picture is the phantom (rectangular box with a hemisphere on
top surface) which, in particular, contains a hidden inclusion
inside (not visible in this photo). The hemisphere mimics the mouse
head in animal experiments of stroke studies, and hidden inclusions
of ink-mix mimic blood clots. The four needles are laser fibers that
provide light sources in our experiments. The fiber on the right
hand side of Fig. \ref{F2} can be moved to three other positions
which are source positions 1,2,3 on Fig. \ref{F1a}. Three other
fibers represent positions 4,5,6 of the source on Fig. \ref{F1a}. A
CCD camera is mounted directly above the setup (not shown in photo),
and the camera focus is on the top surface of the phantom. CCD
stands for Charge Coupled Device. CCD camera with its sensitivity
and response range is commonly used for near infrared laser imaging
of animals.

\begin{figure}[th]
\begin{center}
\begin{minipage}{130mm}
\subfigure[]{ \resizebox*{6cm}{!}{\label{F3a}\includegraphics{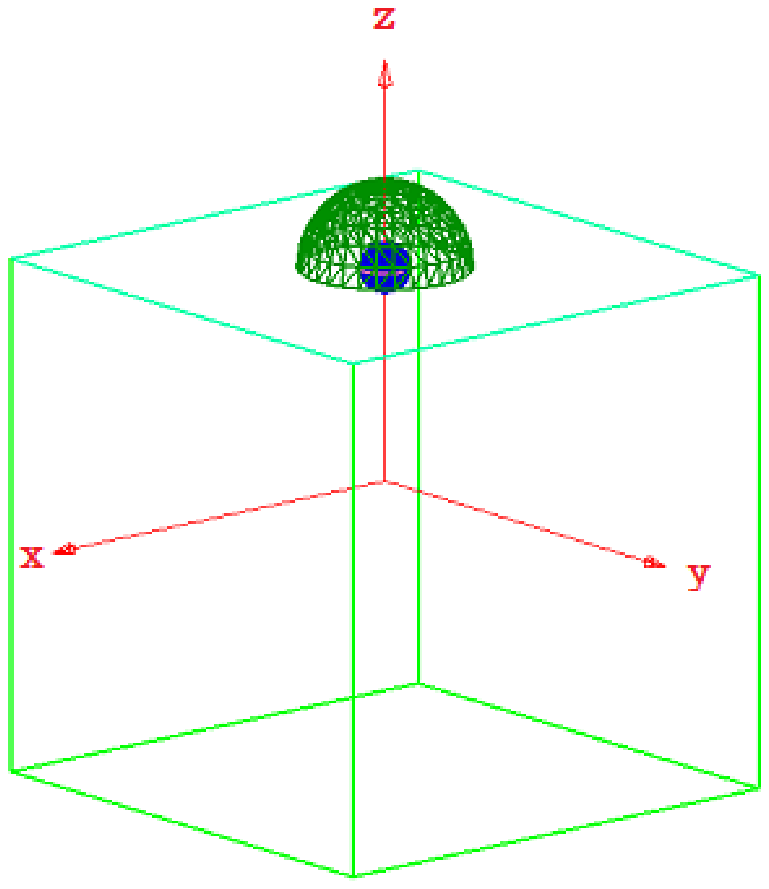}}} %
\subfigure[]{ \resizebox*{6cm}{!}{\label{F3b}\includegraphics{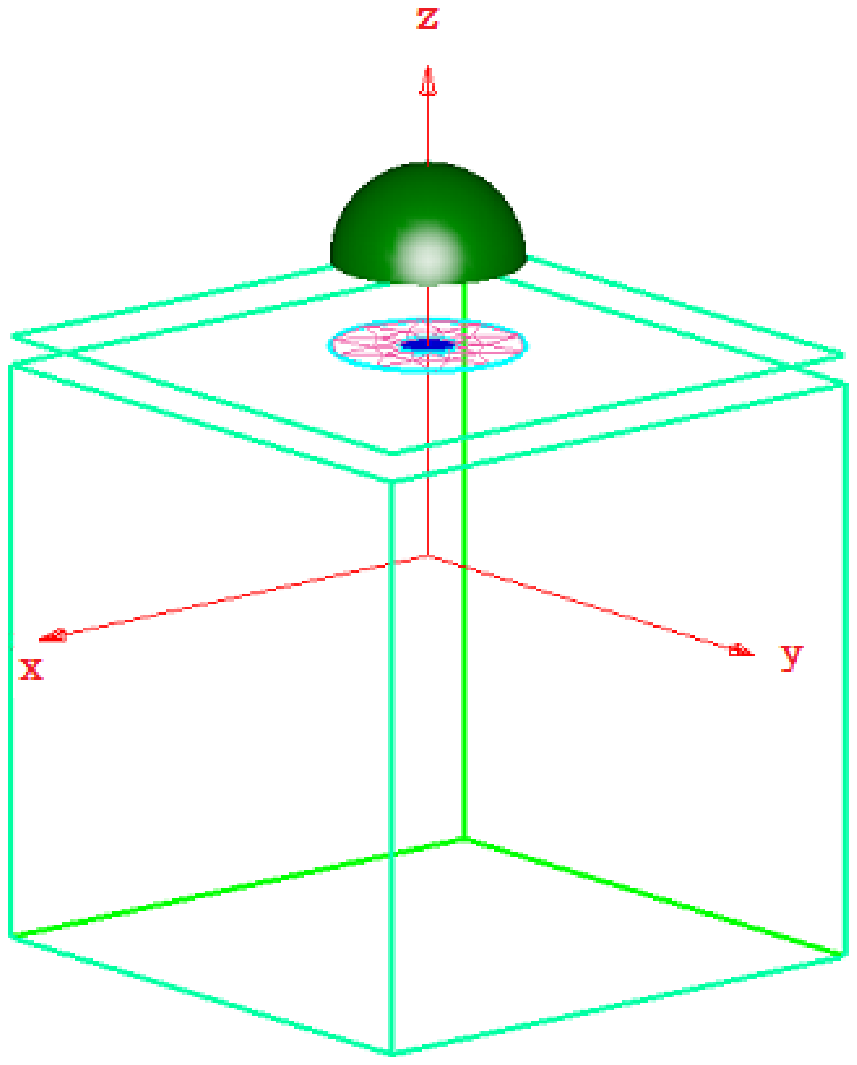}}} %
\caption{\emph{The schematic diagrams of our data acquisition
process. (a) schematically depicts a 3-D phantom and its hidden
inclusion (shown in meshed surface) 5mm below the phantom. The light
source is placed at various locations on the top surface of the
rectangular block (see Fig. \ref{F1a}) and light intensity
measurements are taken on the top surface of phantom. (b) The
measurement surface by a CCD camera. The data are collected from the
surface of the hemisphere as well as from the un-shaded area of top
rectangle, both area are lifted from the phantom in this drawing for
a better illustrative purpose. The rectangular figure in Fig.
\ref{F3b} also illustrates the middle a 2-D cross-section (meshed
circle) of the presumed \textquotedblleft animal head", at the
boundary of which light intensity data are collected for the
reconstruction. The light intensity at the 2-D cross-section (except
for the boundary) is obstructed by the top surface of the
hemisphere. Light sources are also located in the same plane as this
cross-section area. The 2-D inverse problem is solved in this
cross-section by ignoring the dependence on the orthogonal
coordinate. }} \label{F3}
\end{minipage}
\end{center}
\end{figure}

The purpose of this experimental setup is to study the feasibility
of using our numerical method in an animal stroke model. In the case
of a real animal, the top hemisphere (meshed shape in Fig. 6.2a) is
to be replaced by a mouse head and the rectangular block of phantom
to be replaced by an optical mask filled with a matching" fluid (a
tissue-like solution with the optical properties similar to the
animal skin/skull). The image reconstruction should provide the
spatial distribution of the optical coefficient $a(\mathbf{x})$
(directly related to the blood content) in a 2-D cross-section of
the animal brain.

The goal to is to accomplish a noninvasive imaging means to monitor
and investigate hemodynamic dysfunction during ischemic stroke in
animal models. The current experimental setup works only for small
animals, such as rats, but not for human brains due to the limited
penetration depth of NIR light through a larger size and thicker
bone structure of the skull. The methodology developed here,
however, may be applicable to a neonate head, while further studies
are needed to confirm such an expectation.

Our inverse reconstruction is performed in the 2-D plane depicted in Fig. %
\ref{F3b}, with the optical parameter distribution of the medium
inside the circle (at the center of Fig. \ref{F3b}) as an unknown
coefficient in the photon diffusion model. Should we need a
diagnosis of a different cross-section of mouse brain in animal
experiments, we need a different optical mask filled with matching
fluid, and repeat both the experiment and the reconstruction.

\begin{figure}[th]
\centering
\includegraphics[width=0.4\textwidth]{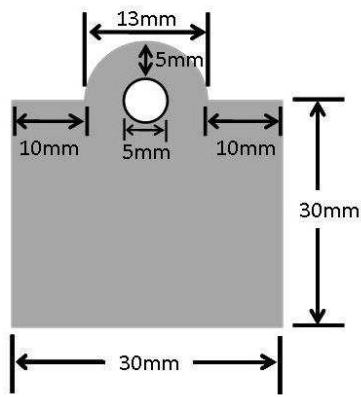}
\caption{\emph{Dimensions of the phantom, shown in a vertical
cross-section at the center of the phantom from side view. The
circle in Fig. \ref{F1a} corresponds to the boundary of the meshed
circle in Fig. \ref{F3b}. This also corresponds to the boundary of
the 13 mm diameter circle at a top view above the rectangular box
(but not related to circle or semicircle in this graph). The 5 mm
diameter hollow sphere is for the placement of inclusions, filled
with liquids of different compositions to emulate strokes. The
hidden inclusion is 5 mm below the top surface.}} \label{F4}
\end{figure}

The geometry of the phantom (shown in a vertical central
cross-section) is depicted in Fig. \ref{F4} with dimensions
specified. The phantom is shaped by a hemisphere (diameter 13 mm) on
the top of a cube of 30mm $\times $ 30mm $\times $ 30mm. A spherical
hollow of 5 mm diameter is located inside the phantom, with its top
5mm below the top surface (Fig. \ref{F4}). The location of this
hollow is symmetric here but can be in any place when needed. One
and two spherical hollows of 3mm diameter inside the phantom are
also used in experiments (see Section 8 for detail). We fill the
hollow with liquids made of different kinds of ink/intralipid mix to
model strokes by blood clots. Intralipid is a product for fat
emulsion, which mimics the response of human or animal tissue to
light at wavelengths in the red and infrared ranges. The phantom is
made of gelatin mixed with the intralipid. The percentage of
intralipid content is adjustable. So that the phantom has the same
optical parameters as the background medium of the target animal
model. At the location of the inclusion (the hollow), we inject
ink/intralipid mixed fluids whose optical absorption coefficients
are 2 times, 3 times and 4 times higher than in the background.
Also, we use the pure black ink to test our reconstruction method
for the case of the infinite absorption. Different levels of the
inclusion/background absorption ratio are used to validate our
method for its ability for different blood clots. Light intensity
measurements are taken directly above the semicircle in Fig.
\ref{F4}. We note that Fig. \ref{F1a} shows a top view of the layout
but Fig. \ref{F4} shows a cross-section from side view. The disk region $%
\Omega $ in Fig. \ref{F1a} corresponds to the meshed area in Fig.
\ref{F3b} as well as a horizontal cross-section of the 13 mm
hemisphere. The reason we use such a geometry is that we need to
test our capability to reconstruct inclusions at several different
depth (moving up and down vertically as in Fig. \ref{F4}) as well as
at different locations (moving horizontally in Fig. \ref{F3b}).

\section{Processing Real Data}

As it is typical when working with imaging from real data, we need
to make several steps of data pre-processing before applying our
inverse algorithm.

\subsection{Computational domains}

We need to use several computational domains. The computational
domains and meshes used in our numerical calculation involve four
domains. Fig. \ref{F5a} - Fig. \ref{F5c} and Fig. \ref{F1b} are the
finite element meshes on each of these domains $\Omega ,\Omega
_{0},\Omega _{0}\diagdown \Omega ,\Omega _{1}$ respectively. In
actual computations we use more refined meshes for each domain than
those illustrated. Therefore, figures are not scaled to actual
sizes.

\begin{figure}[th]
\begin{center}
\begin{minipage}{130mm}
\subfigure[$\Omega$]{     \resizebox*{4cm}{!}{\label{F5a}\includegraphics{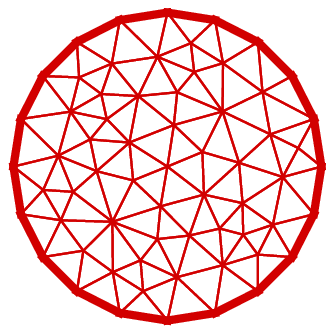}}} %
\subfigure[$\Omega_{0}$]{ \resizebox*{4cm}{!}{\label{F5b}\includegraphics{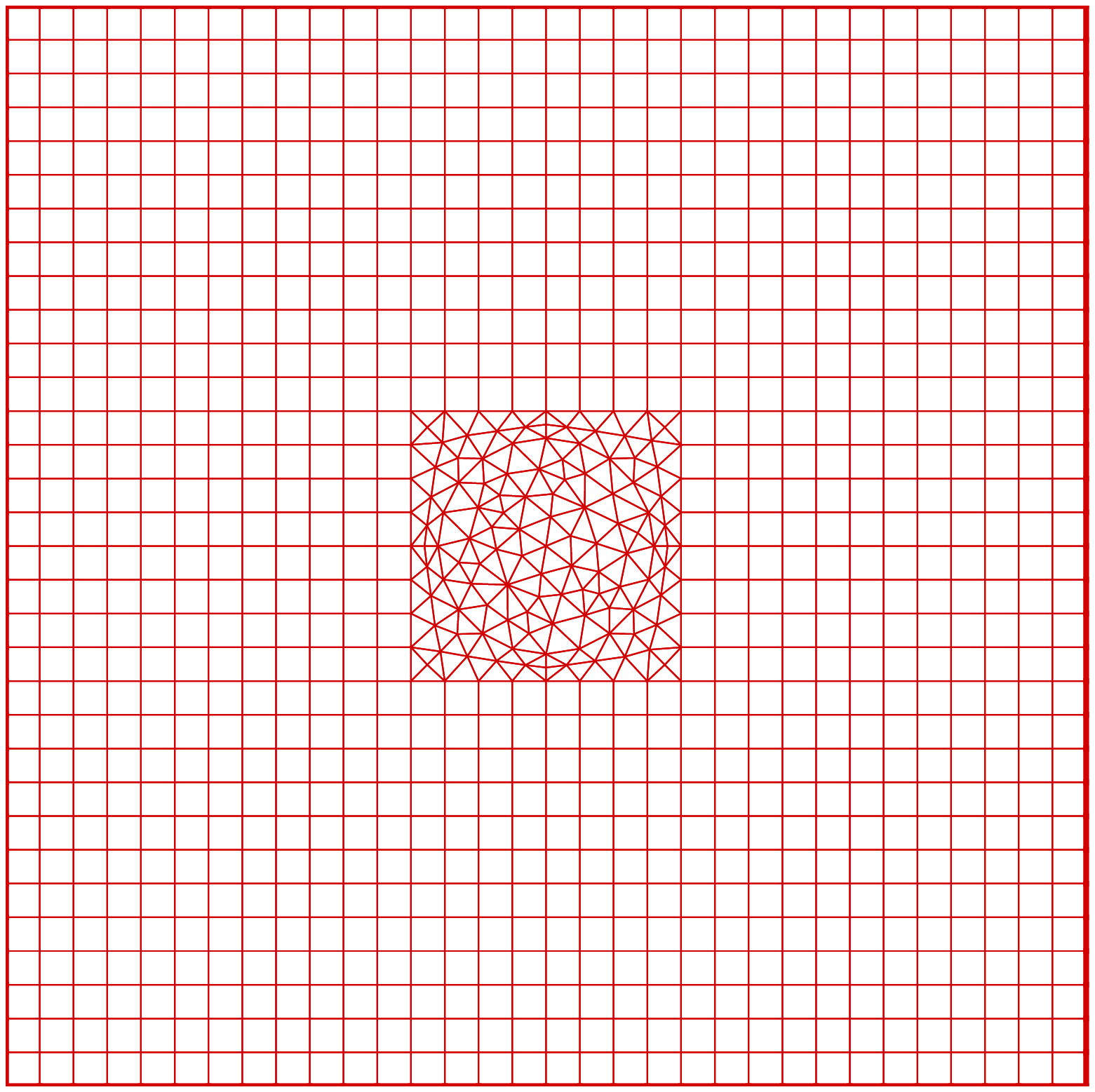}}} %
\subfigure[$\Omega_{0}\backslash\Omega $]{ \resizebox*{4cm}{!}{\label{F5c}\includegraphics{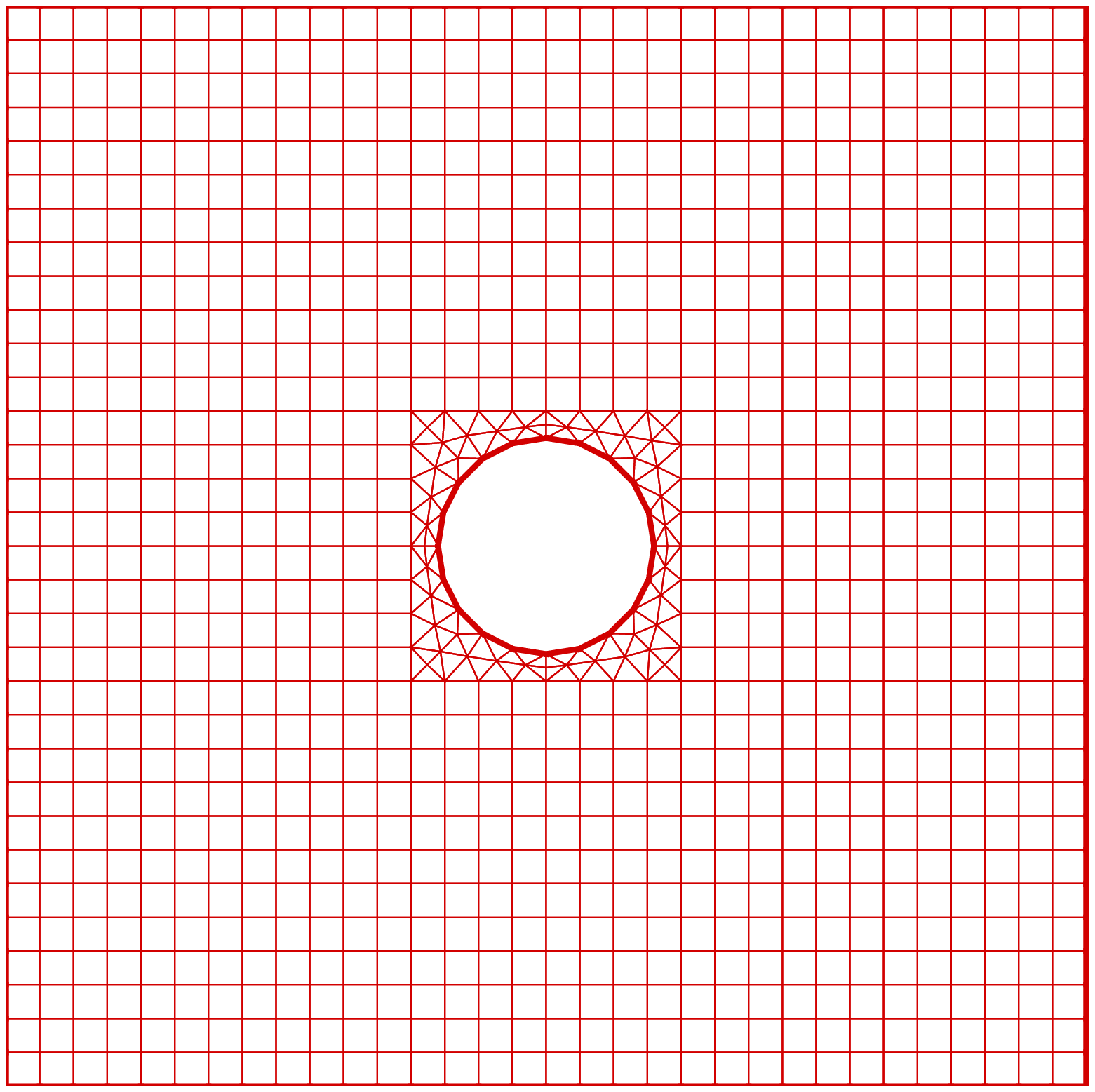}}} %
\caption{\emph{(a) The domain $\Omega
=\{(x,z),\protect\sqrt{x^{2}+z^{2}}<4.63\}$ (units: mm) with a
triangular mesh. The domain represents a cross-section of our
phantom, and the inclusion inside $\Omega $ is to be reconstructed.
(b) The domain $\Omega _{0}=\{(x,z),|x|<23.32,|z|<23.32\}$ (units:
mm). The domain $\Omega_{0}$ contains $\Omega$ and light sources.
Meshed small square is the domain $\Omega _{1}$ which is the same as
on Fig. \ref{F1b}. Computations of the inverse problem are performed in $%
\Omega _{1}.$ (c) Shows $\Omega _{0}\diagdown \Omega $ the domain
used for
data processing. To smooth out the measurement noise, equation (\protect\ref%
{2.1}) with the Dirichlet boundary conditions $u\mid _{\partial \Omega }=%
\protect\varphi \left( x,x_{0}\right) $, $u\mid _{\partial \Omega
_{0}}=0$
is solved in $\Omega _{0}\diagdown \Omega $. Then the smoothed data along $%
\partial \Omega _{1}$ is used for the inverse problem.}}
\label{F5}
\end{minipage}
\end{center}
\end{figure}

Here the disk $\Omega $ is the domain of interest that contains the
inclusion, corresponding to the meshed area in Fig. \ref{F3b}. It
reflects a cross-section of the hemisphere of phantom in Fig.
\ref{F3b}. The light intensity data used in our computations are
originated by the measurements at $\partial \Omega $ by taking
pictures with CCD camera on the top of phantom as discussed, shown
in Fig. \ref{F3b}. The CCD camera collected data is acquired from
the surface of the hemisphere and the un-shaded area of the top
rectangle, and we extract only the needed data for $\partial \Omega
$ i.e., the circle in Fig. \ref{F3b} by getting the data at these
locations as boundary values. Here $\Omega =\left\{ \left(
x,z\right) |\sqrt{x^{2}+z^{2}}\leq 4.63mm\right\} $.

On Fig. \ref{F5} $\Omega _{0}$ is a large domain, which can be
interpreted as a truncated plane $\mathbb{R}^{2}.$ Indeed, one
cannot practically solve
the forward problem (\ref{2.1}), (\ref{2.2}) in the infinite plane $\mathbb{R%
}^{2}.$ Light sources are located in $\Omega _{0}\diagdown \Omega $.
The background simulation and calibration of background parameters
are performed in $\Omega _{0}.$ To smooth out the noise in the data,
equation (\ref{2.1}) is solved in $\Omega _{0}\diagdown \Omega $ for
each source position, which is similar with \cite{bib22}. Because of
(\ref{2.3}), we use $a\left( \mathbf{x}\right) :=k^{2}$ for $x\in
\Omega _{0}\diagdown \Omega .$ In this procedure the Dirichlet
boundary condition at $\partial \Omega $ is taken
from the real data, and we use the zero Dirichlet boundary condition at $%
\partial \Omega _{0}.$ As a result, we obtain the smoothed Dirichlet
boundary condition at $\partial \Omega _{1}$ for each source
location. We solve the inverse problem in $\Omega _{1},$ and $\Omega
\subset \Omega _{1}$.

\subsection{Data pre-processing and approximations of boundary conditions}

As it was stated in subsection 7.1, to smooth out our measurement
data, we solve the Dirichlet boundary problem for equation
(\ref{2.1}) with $a\left( \mathbf{x}\right) \equiv k^{2}$ for
$\mathbf{x}\in \Omega _{0}\diagdown \Omega $ (Fig. \ref{F5c}) for
each source location. Namely,
\begin{eqnarray}
\Delta u-k^{2}u &=&-\delta \left( \mathbf{x}-\mathbf{s}^{\left(
i\right)
}\right) ,\mathbf{x}\in \Omega _{0}\diagdown \Omega ,  \label{7.1} \\
u\mid _{\partial \Omega } &=&\varphi \left(
\mathbf{x},\mathbf{s}^{\left( i\right) }\right) ,\text{ }u\mid
_{\partial \Omega _{0}}=0.  \notag
\end{eqnarray}%
Here $k^{2}=3\mu _{s}^{\prime }\mu _{a}$ is the background value and $%
\varphi \left( \mathbf{x,}\mathbf{s}^{\left( i\right) }\right) $ is
the experimentally measured data for the light source
$\mathbf{s}^{\left( i\right) }$. Let $\overline{\varphi }\left(
\mathbf{x},\mathbf{s}^{\left( i\right) }\right) $ be the trace of
the solution of the Dirichlet boundary value problem (\ref{7.1}) at
$\mathbf{x}\in \partial \Omega _{1}$. We solve
the inverse problem in the square $\Omega _{1}$ with the smoothed data $%
u\mid _{\partial \Omega _{1}}=\overline{\varphi }\left( \mathbf{x},\mathbf{s}%
^{\left( i\right) }\right) .$ We have shown in \cite{bib23,bib25}
that solving inverse problems in the original domain $\Omega $ and
extended domain $\Omega _{1}$ are equivalent mathematically. But
numerically the noisy component in $\overline{\varphi }\left(
\mathbf{x},\mathbf{s}^{\left(
i\right) }\right) $ is much smaller than in $\varphi \left( \mathbf{x},%
\mathbf{s}^{\left( i\right) }\right) .$ This smoothing effect takes
place because the inverse problem is solved in $\Omega _{1}$.

\subsection{The forward problem and calibration}

We now address the question on which value of $k^{2}$ one should use
when working with the real data. The optical properties of the
background of the phantom (without the hidden inclusion) are known
theoretically from the concentration of intralipid in the mix.
However, there is a discrepancy between the theoretical value and
actual measurements. Before we solve the inverse problem to image
hidden inclusions, we calibrate our model by adjusting the
background value of $k^{2}$ to the real data measured for the
reference medium, which is the phantom without inclusion, i.e. the
hollow is filled with the same intralipid solution as that in the
phantom itself.

First, we numerically solve the forward problem with the source position $%
\mathbf{s}^{\left( 1\right) }$ in the domain $\Omega _{0}$ without
any inclusion,
\begin{eqnarray}
\Delta u-k^{2}u &=&-A\delta \left( \mathbf{x}-\mathbf{s}^{\left(
1\right)
}\right) ,\mathbf{x}\in \Omega _{0},  \label{7.2} \\
u\mid _{\partial \Omega _{0}} &=&0,  \notag
\end{eqnarray}%
where $k^{2}=3\mu _{s}^{\prime }\mu _{a}$ is the background value.
Then we calibrate the parameter $\mu _{a}$ (fixing $\mu _{s}^{\prime
}$) as well as the amplitude $A>0$ of light source in our model
(\ref{7.2}) to match the measured light intensity for
$\mathbf{s}^{\left( 1\right) }$ for the uniform background. We do
not know the number $A$. Thus, we choose $A$ in such a way that
$u_{comp}\left( \mathbf{x}_{\max },\mathbf{s}^{\left( 1\right)
}\right) \approx u_{meas}\left( \mathbf{x}_{\max
},\mathbf{s}^{\left( 1\right) }\right) .$ Here $\mathbf{x}_{\max }$
is the brightest point, i.e. the far
right point on $\partial \Omega $, being closest to the light source. Also, $%
u_{comp}\left( \mathbf{x}_{\max },\mathbf{s}^{\left( 1\right)
}\right) $ and $u_{meas}\left( \mathbf{x}_{\max },\mathbf{s}^{\left(
1\right) }\right) $ are computed and measured light intensities
respectively. Next, we should approximate the constant $k^{2}.$ To
do this, we take another sampling point $\mathbf{x}_{\min }$ with
the minimal light intensity, which is the farthest left point on
$\partial \Omega $. Next we consider ratios $R_{comp}\left(
k^{2}\right) ,R_{meas},$ where
\begin{equation*}
R_{comp}\left( k^{2}\right) =\frac{u_{comp}\left( \mathbf{x}_{\max },\mathbf{%
s}^{\left( 1\right) }\right) }{u_{comp}\left( \mathbf{x}_{\min },\mathbf{s}%
^{\left( 1\right) }\right) },\text{ }R_{meas}=\frac{u_{meas}\left( \mathbf{x}%
_{\max },\mathbf{s}^{\left( 1\right) }\right) }{u_{meas}\left( \mathbf{x}%
_{\min },\mathbf{s}^{\left( 1\right) }\right) }.
\end{equation*}%
These ratios are independent on the number $A$ in (\ref{7.2}). We choose $%
k^{2}$ such that $R_{comp}\left( k^{2}\right) \approx R_{meas}.$ As
a result, the calibrated value of $k^{2}$ was $k^{2}=2.403.$ This
computed value matches quite well the theoretical value of 2.4 of
the intralipid solution we have used.

\section{ Results Of the Reconstruction}

Let $a_{incl}=a\left( \mathbf{x}\right) $ be the value of $a\left( \mathbf{x}%
\right) $ inside the inclusion and $a_{b}=k^{2}=2.403$ be the value
of the coefficient $a\left( \mathbf{x}\right) $ in the background,
which was computed in subsection 7.3. Our ratios $a_{incl}/a_{b}$
where
\begin{equation}
\frac{a_{incl}}{a_{b}}=2,3,4,\infty .  \label{8.1}
\end{equation}%
The value $a_{incl}/a_{b}=\infty $ means that the inclusion was
filled with a black absorber, i.e. black ink. Table \ref{Table1}
gives a detailed information of our real data with different
positions, numbers, diameters and contrast of inclusions. "Y" means
that the experiment was done for this setting, and "-" means that
the experiment was not done. Note that in Group 3 we have imaged two
different inclusions simultaneously.

\begin{table}[htbp]
\begin{center}
\begin{tabular}{|c|c|c|c|c|c|c|}
\hline
Inclusion Groups & Positions & Diameters & Ratio2 & Ratio3 & Ratio4 & Ratio$%
\infty$ \\ \hline Group 1 &
\includegraphics[width=0.45in]{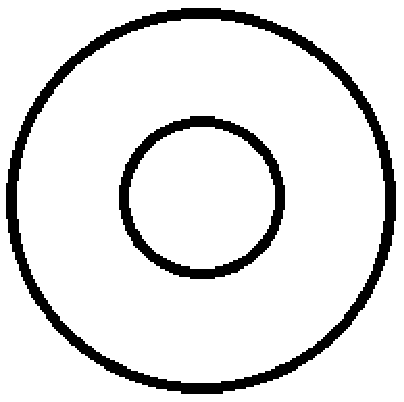} & 5mm & Y & Y & Y & Y
\\ \hline
Group 2 & \includegraphics[width=0.45in]{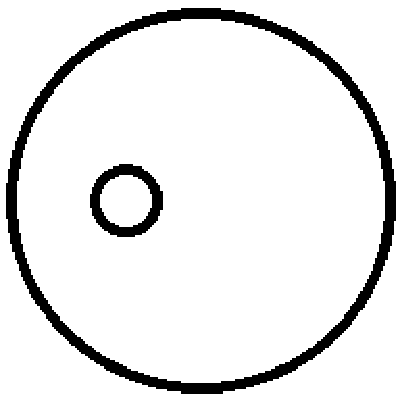} & 3mm & Y & Y
& Y & -
\\ \hline
Group 3 & \includegraphics[width=0.45in]{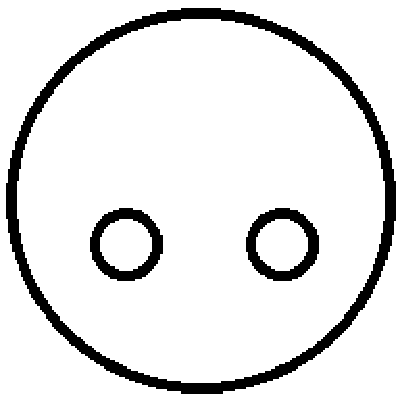} & 3mm & Y & Y
& Y & -
\\ \hline
\end{tabular}%
\end{center}
\caption{Real data} \label{Table1}
\end{table}

As described in sub-subsection 5.2.1, we construct the
\textquotedblleft asymptotic tail\textquotedblright\ on the first
stage using light sources 1,4,5, and 6 as an initial approximation,
to be followed by other
subroutines for further refinements. The image of the function $a(x,z)$ in (%
\ref{5.8}) is depicted on Fig. \ref{F6}, for a phantom where the
theoretical
value of the inclusion/background contrast (\ref{8.1}) was $a_{incl}/a_{b}=3$%
. Images from the asymptotic tails for other phantoms listed in
Table 8.1 were similar.

\begin{figure}[htbp]
\centering
\includegraphics[width=0.4\textwidth]{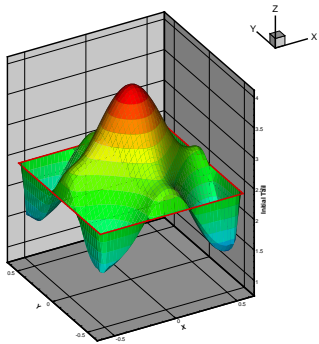}
\caption{\emph{The function $a\left( x,z\right) $\ in
(\protect\ref{5.8}) is depicted. This function was obtained on the
first stage procedure for the tail (sub-subsection 5.2.1) as an
initial approximation. The initial reconstruction is obtained from a
phantom where the theoretical value of the inclusion/background
contrast (\protect\ref{8.1}) was $a_{incl}/a_{b}=3.$}} \label{F6}
\end{figure}

Figures \ref{F10}- Fig. \ref{F12} depict the reconstructed images
from our real data. Fig. \ref{F10} depicts the 3D plots of
reconstructed images of group 1 for contrast values of: 2:1, 3:1,
4:1 and $\infty :1$, from the left to the right respectively. Fig.
\ref{F11} - Fig. \ref{F12} show 2:1, 3:1,
4:1, from left to right respectively for groups 2,3. In all our examples $%
\varepsilon =10^{-5}$ in (\ref{5.10}). The finally reconstructed
results for
contrasts $a_{b}^{-1}\max a\left( \mathbf{x}\right) $ are listed in Table %
\ref{Table2}. Note that our above reconstruction algorithm does not
use any
knowledge of neither the location of the inclusion, nor the contrast value $%
a_{incl}/a_{b}$.

\begin{figure}[htbp]
\centering
\includegraphics[width=0.24\textwidth]{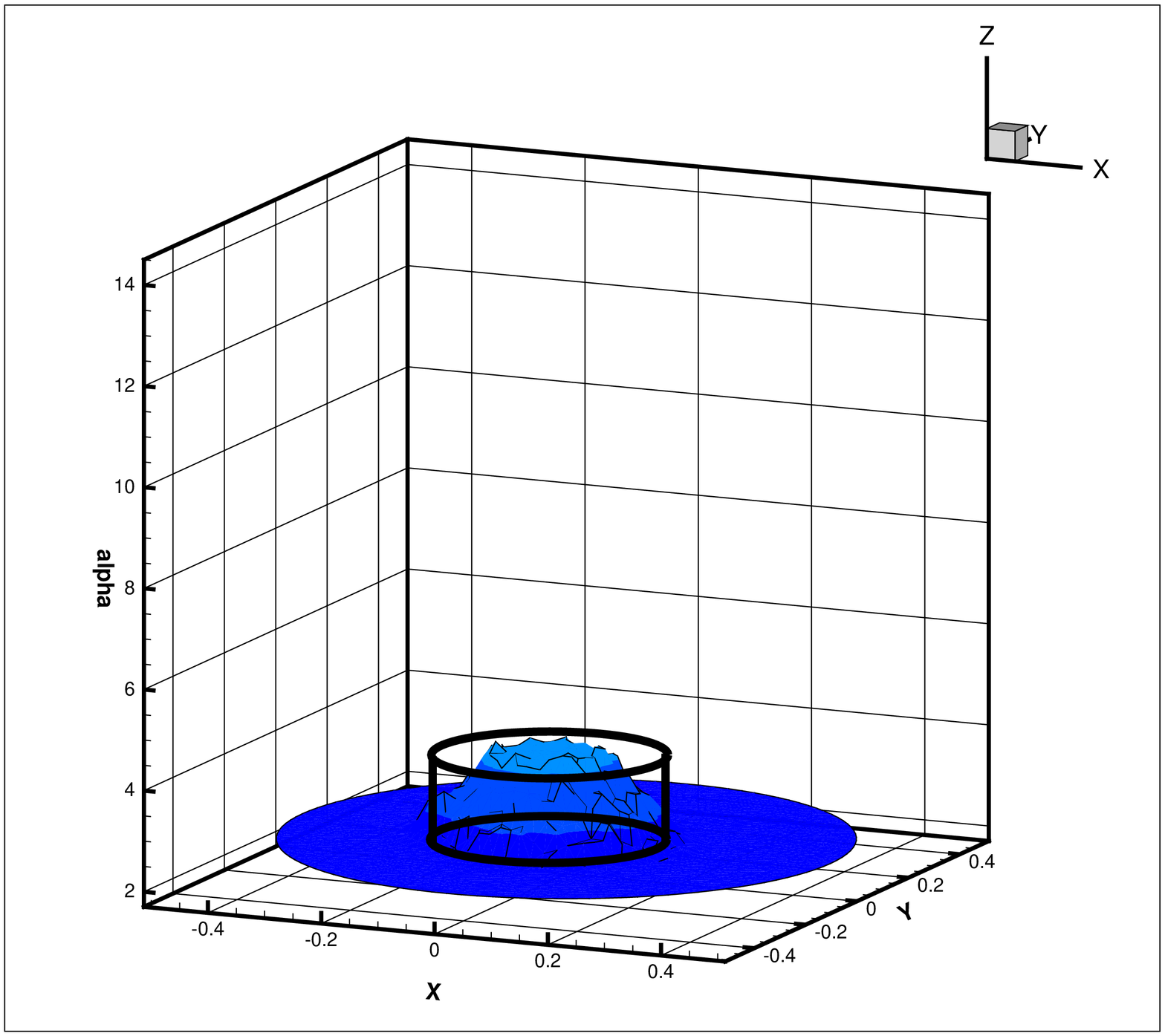} %
\includegraphics[width=0.24\textwidth]{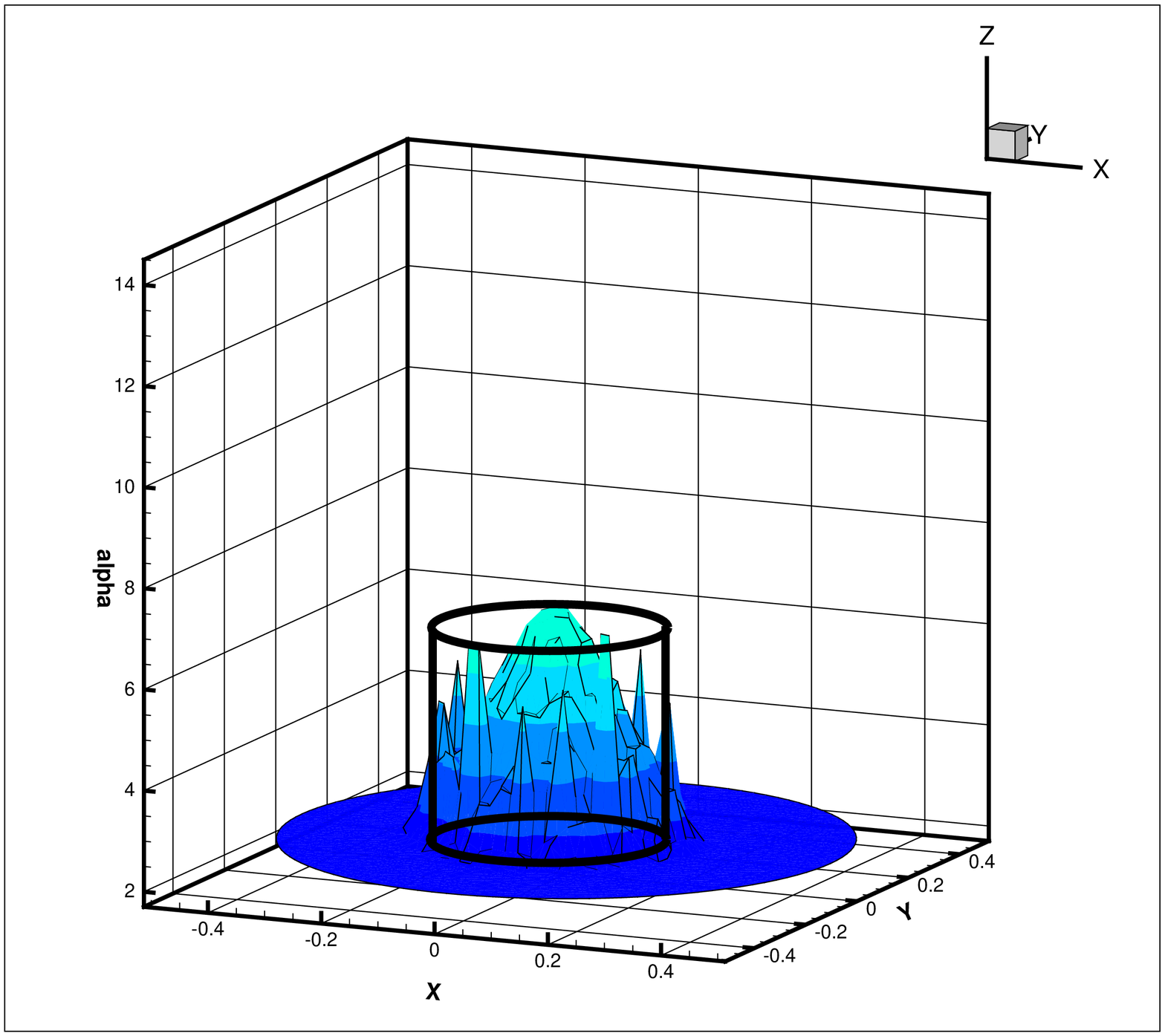} %
\includegraphics[width=0.24\textwidth]{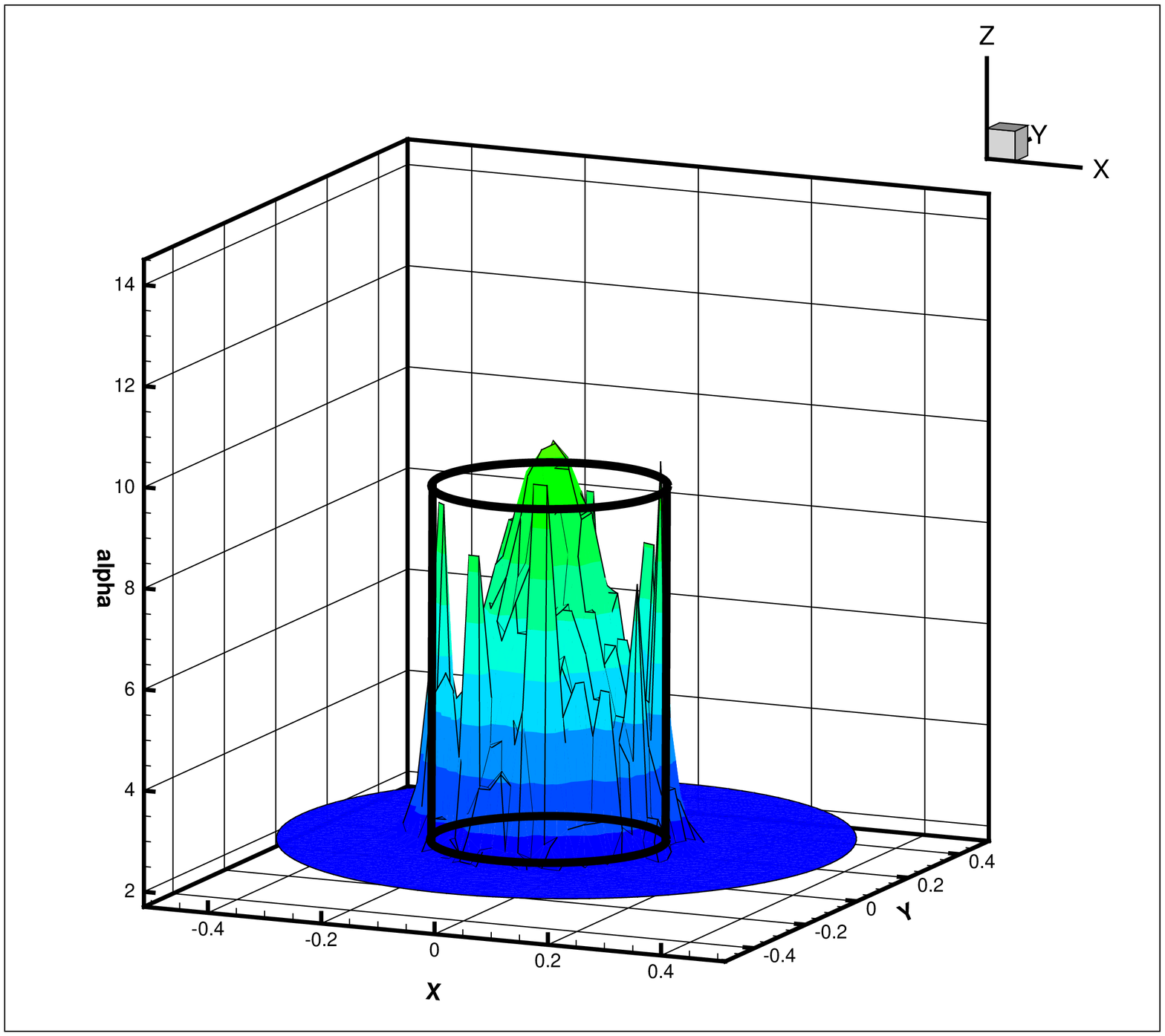} %
\includegraphics[width=0.24\textwidth]{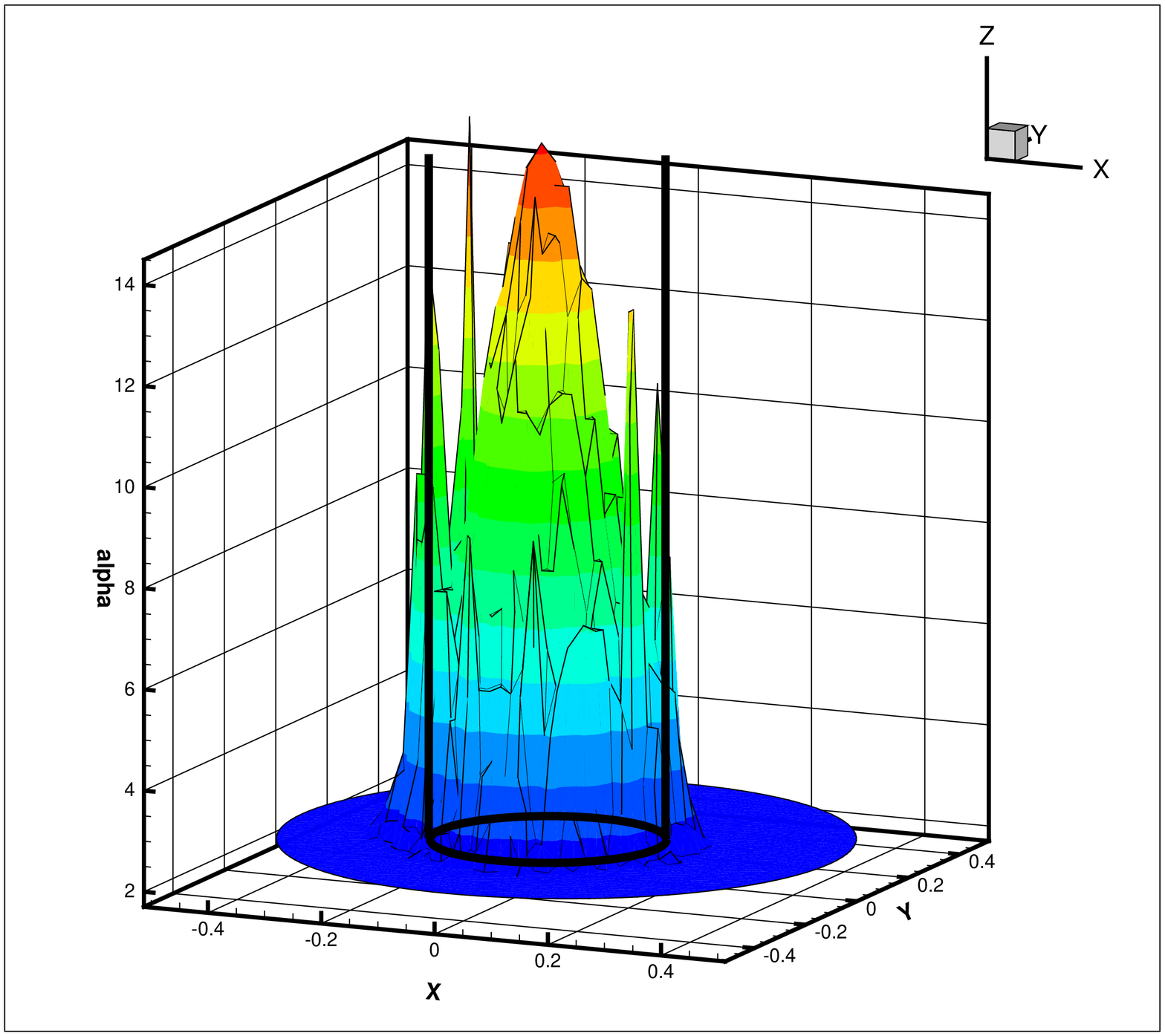}
\caption{Reconstructed inclusion contrast of data group 1, actual
contrasts are 2:1, 3:1, 4:1 and $\infty:1$, from left to right
respectively. The transparent frames show the theoretical values of
inclusion/background contrasts of actual inclusions, which are made
of different ink-intralipid mix.} \label{F10}
\end{figure}

\begin{figure}[htbp]
\centering
\includegraphics[width=0.28\textwidth]{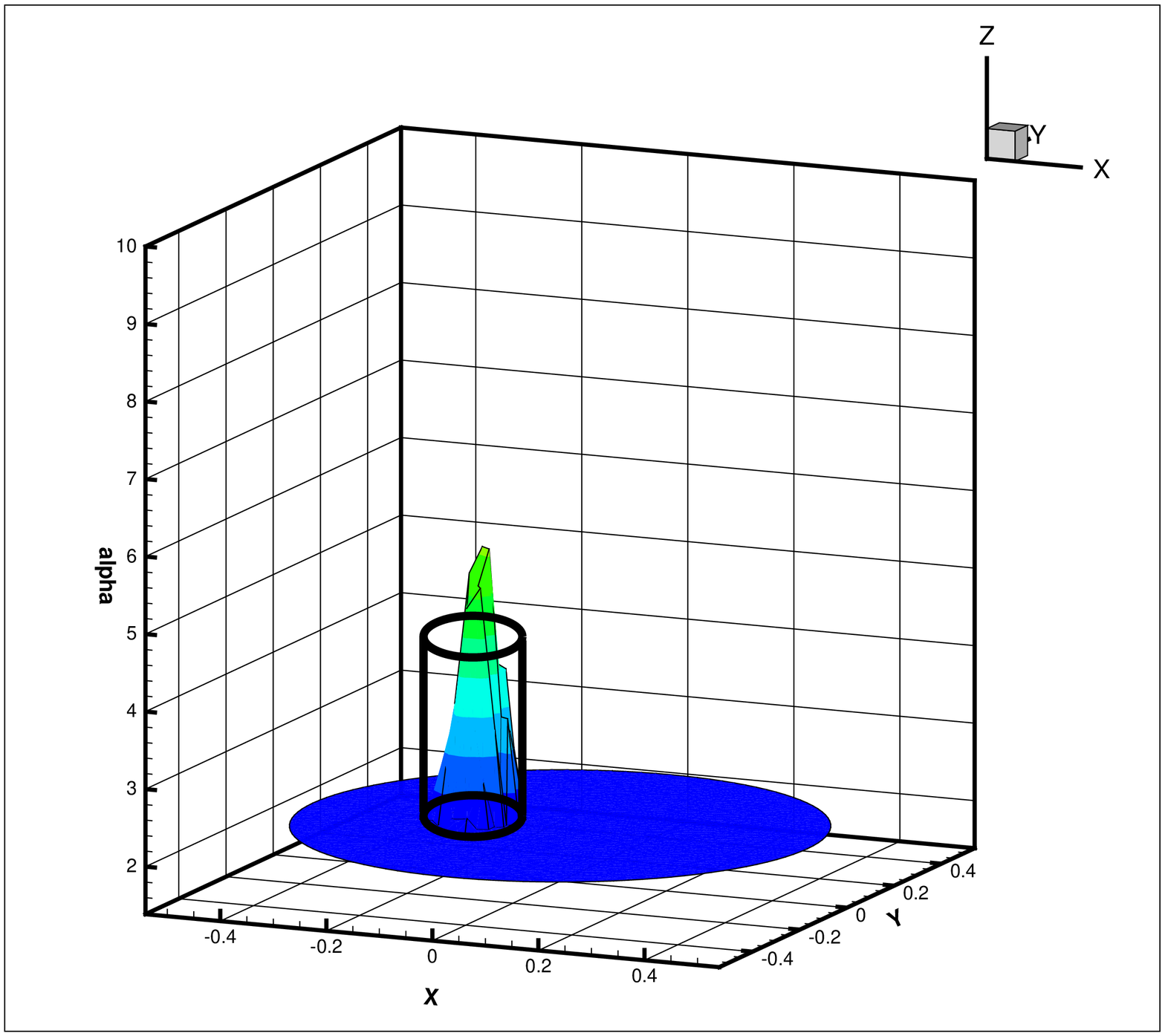} %
\includegraphics[width=0.28\textwidth]{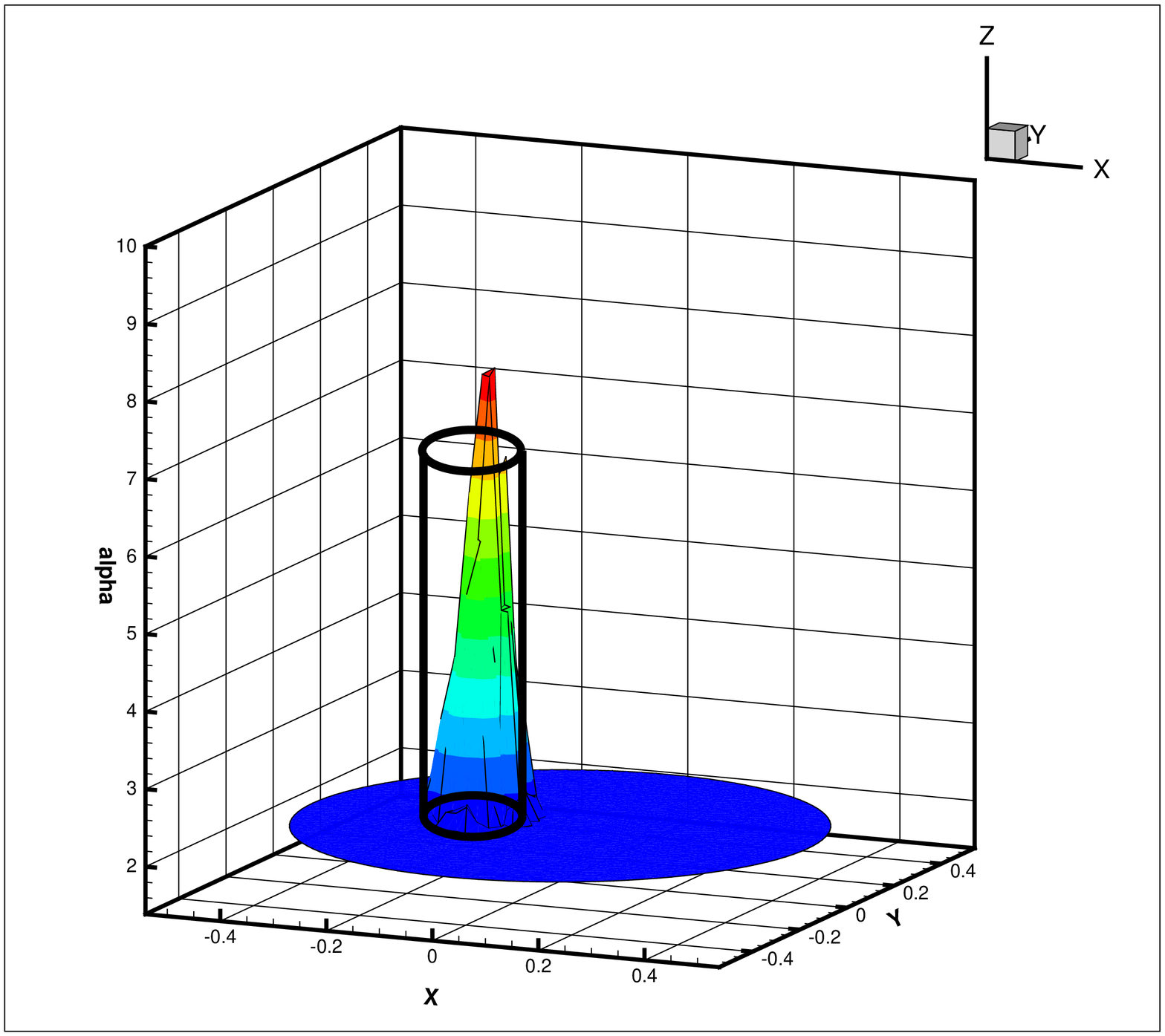} %
\includegraphics[width=0.28\textwidth]{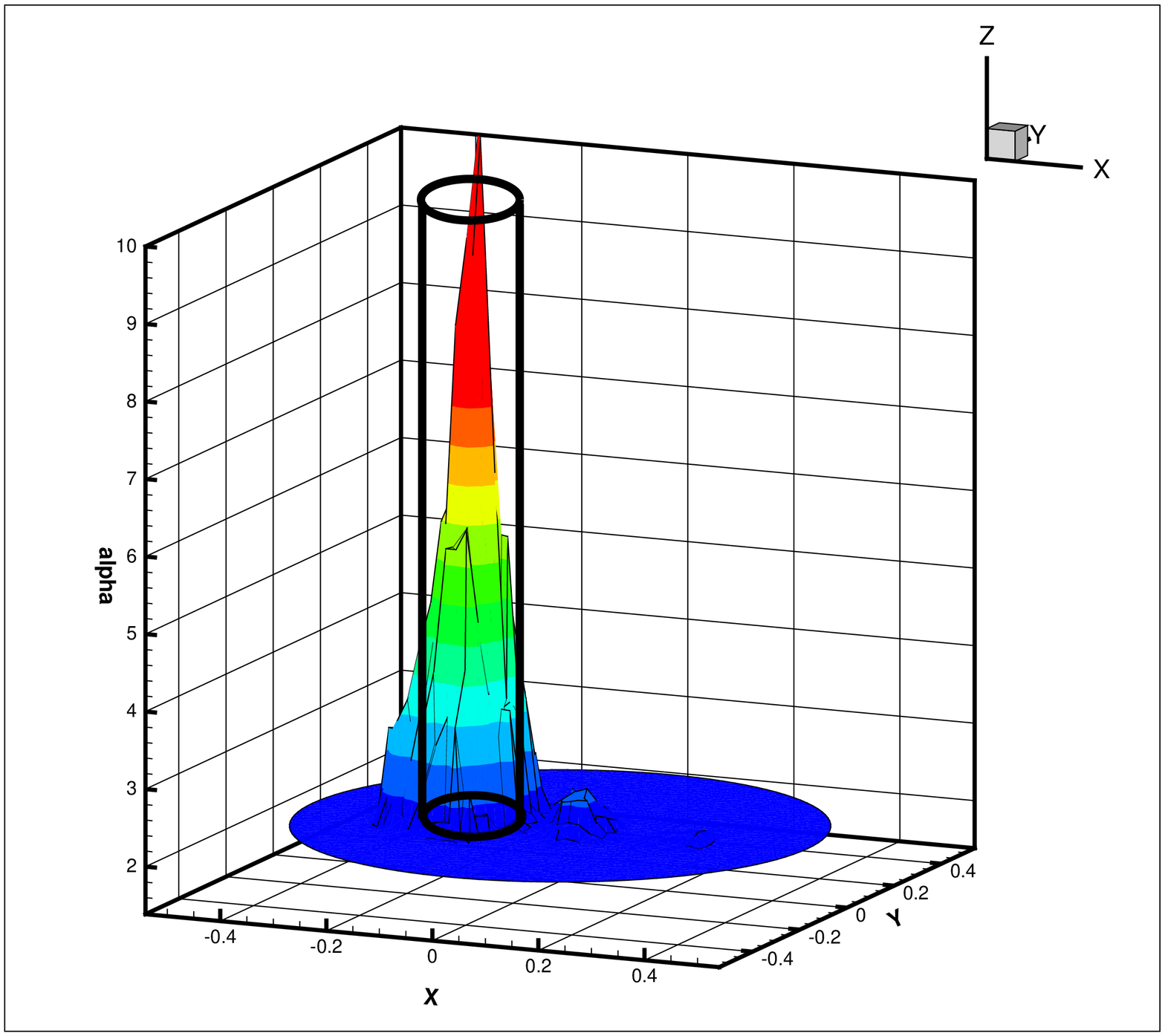}
\caption{Reconstructed inclusion contrast of data group 2, actual
contrasts are 2:1, 3:1, 4:1, from left to right respectively. The
transparent frames show the theoretical values of
inclusion/background contrasts of actual inclusions, which are made
of different ink-intralipid mix. } \label{F11}
\end{figure}

\begin{figure}[htbp]
\centering
\includegraphics[width=0.28\textwidth]{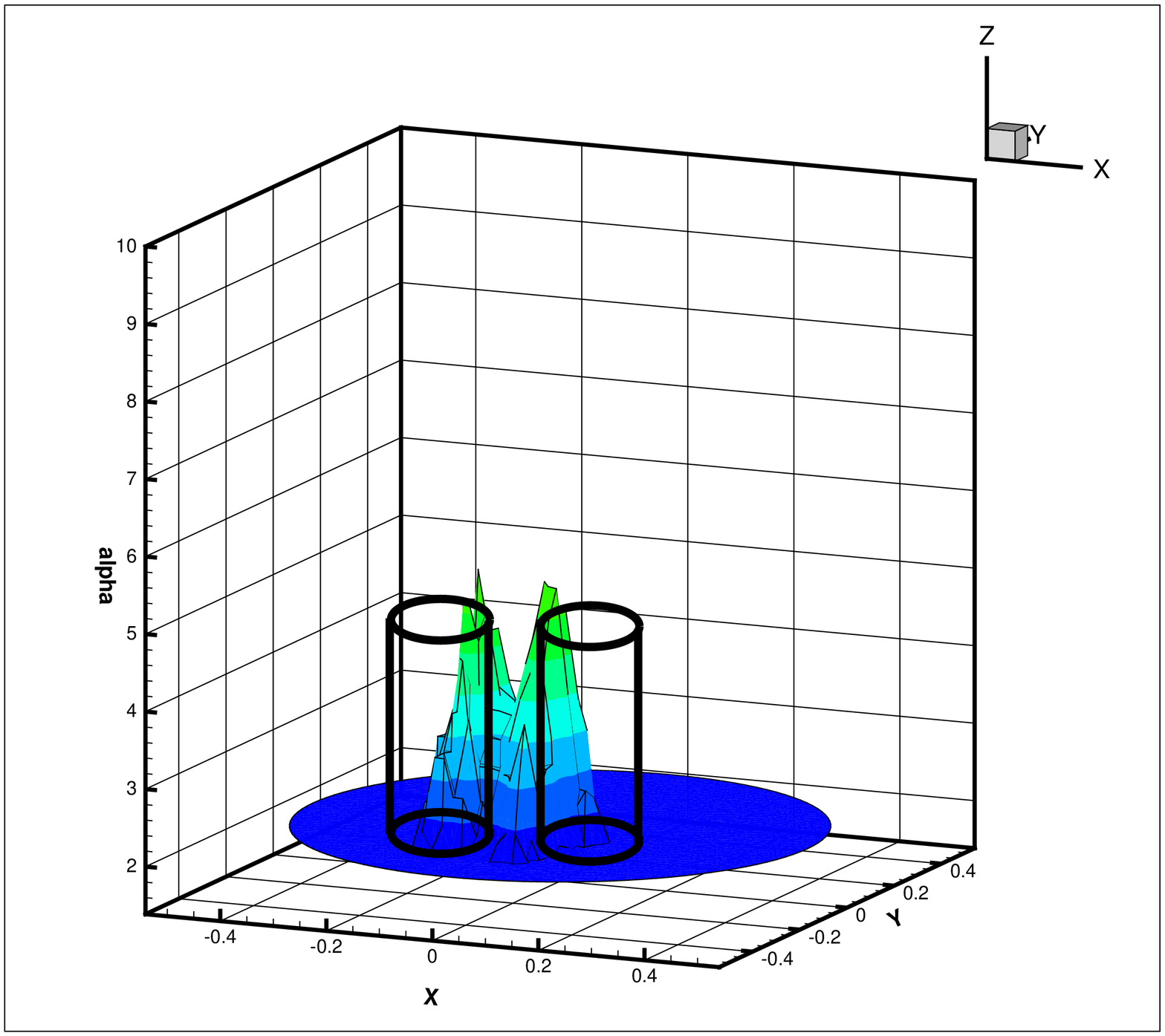} %
\includegraphics[width=0.28\textwidth]{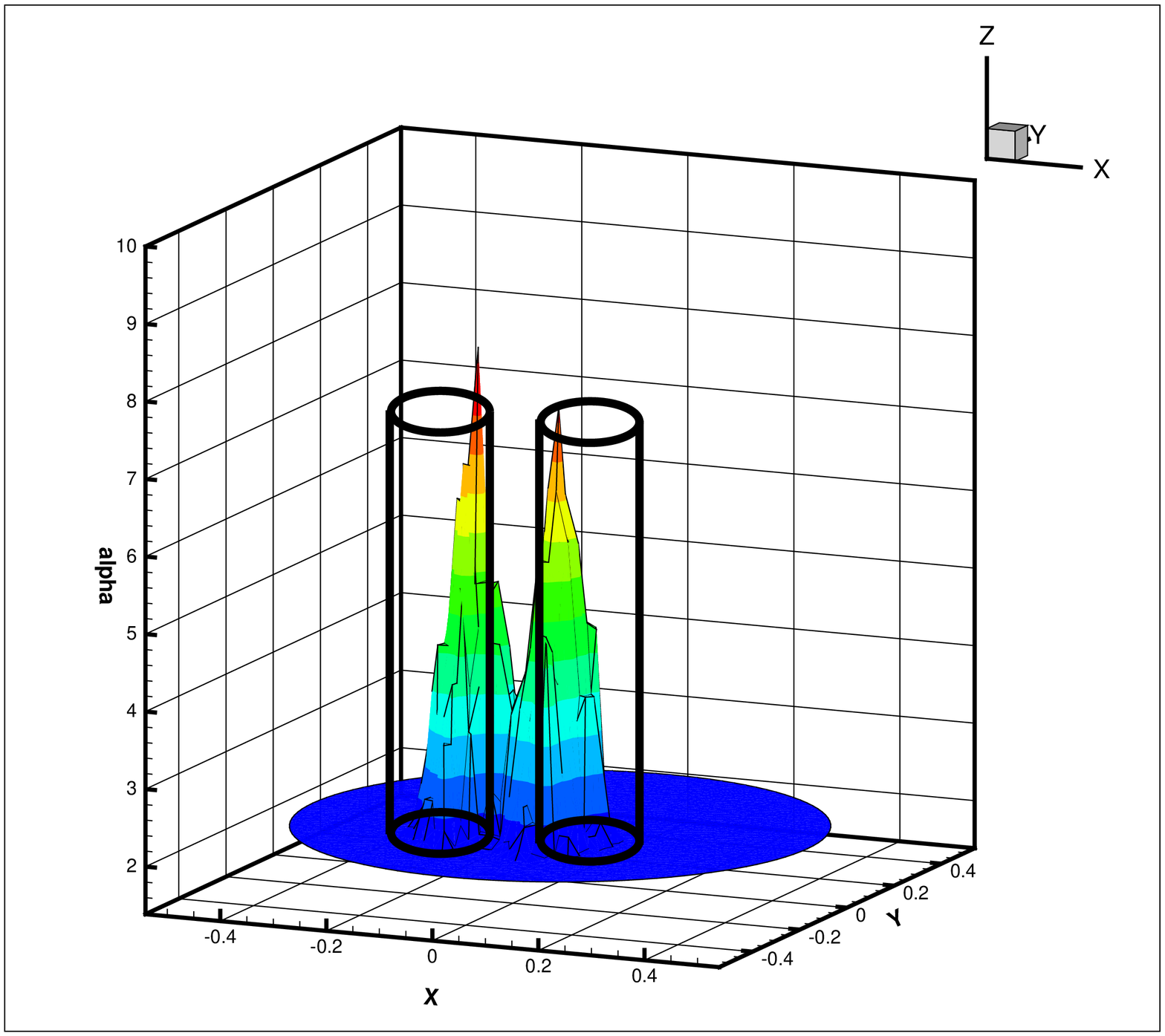} %
\includegraphics[width=0.28\textwidth]{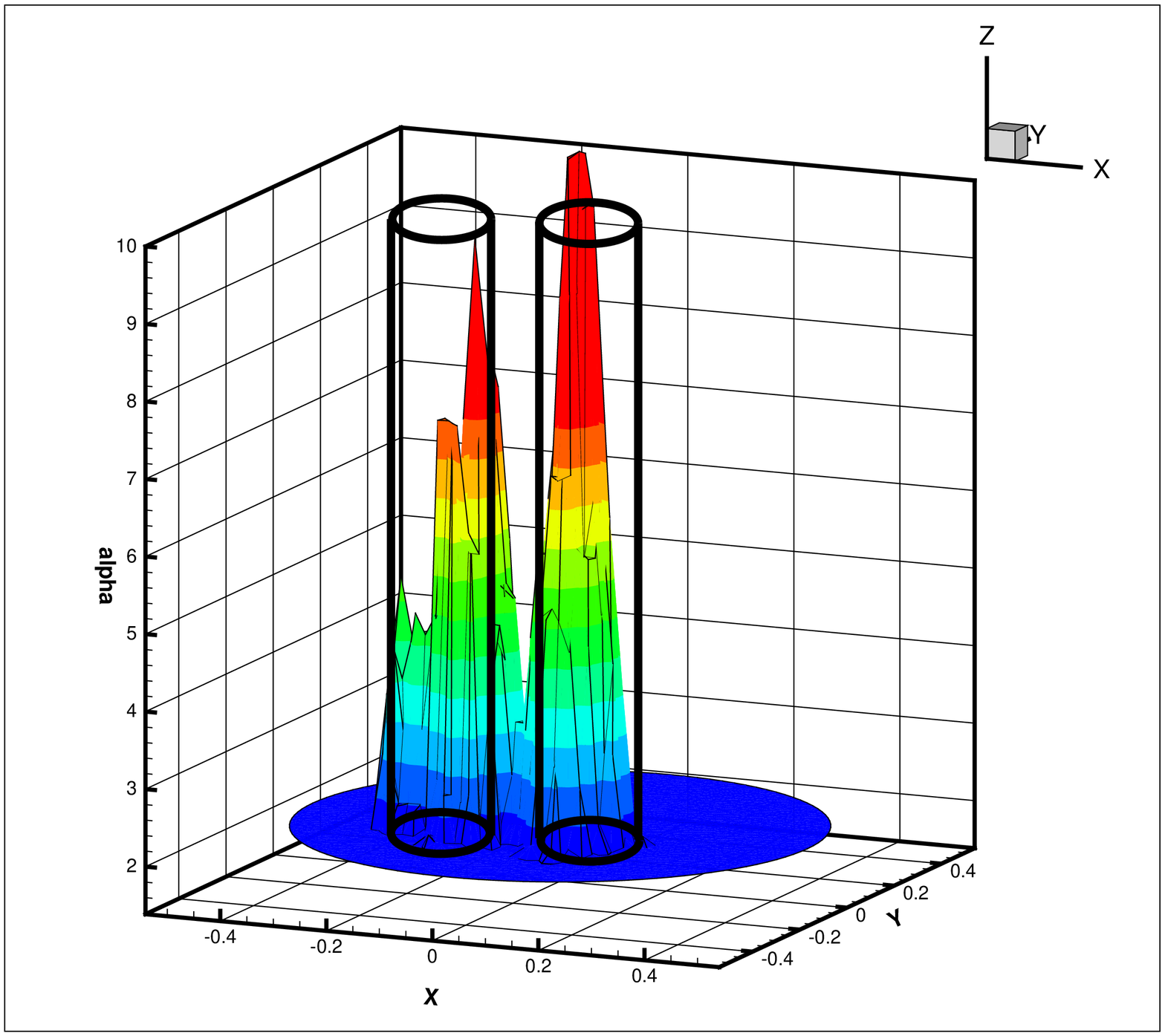}
\caption{Reconstructed inclusion contrast of data group 3, actual
contrasts are 2:1, 3:1, 4:1, from left to right respectively. The
transparent frames show the theoretical values of
inclusion/background contrasts of actual inclusions, which are made
of different ink-intralipid mix.} \label{F12}
\end{figure}

\begin{table}[htbp]
\begin{center}
\begin{tabular}{|c|c|c|c|}
\hline
& The true contrast $a_{incl}/a_{b}$ & $a_{b}^{-1}\max a\left( \mathbf{x}%
\right) $ & Relative Error \\ \hline \multirow{4}{*}{Group 1} & 2 &
2.11 & 0.056 \\ \cline{2-4} & 3 & 2.9 & 0.032 \\ \cline{2-4} & 4 &
4.22 & 0.057 \\ \cline{2-4} & $\infty$ & 6.69 & unknown \\ \hline
\multirow{4}{*}{Group 2} & 2 & 2.33 & 0.167 \\ \cline{2-4} & 3 &
3.29 & 0.0986 \\ \cline{2-4} & 4 & 4.57 & 0.143 \\ \hline
\multirow{4}{*}{Group 3} & 2 & 2.29 & 0.148 \\ \cline{2-4} & 3 &
3.49 & 0.164 \\ \cline{2-4} & 4 & 4.58 & 0.147 \\ \hline
\end{tabular}%
\end{center}
\caption{Reconstructed values of the contrast $a_{b}^{-1}\max
a\left(
\mathbf{x}\right) $ within imaged inclusions and they relative errors, $%
a_{b}=2.403,$ compare with (8.1).} \label{Table2}
\end{table}

\bigskip

\section{Summary}

We have worked with a set of real data, using the approximately
globally convergent numerical method of \cite{bib18} for a
Coefficient Inverse Problem for an elliptic equation.\ These data
mimic imaging of clots in the heads of a mouse. We have introduced a
new concept of the approximate global convergence property and have
established this property for the discrete case of a finite number
of finite elements, unlike the continuous case of \cite{bib18}.
Figures \ref{F10}- Fig. \ref{F12} as well as Table \ref{Table2}
demonstrate that our are quite accurate, including both
inclusion/background contrasts and locations of inclusions. Note
that accurate values of contrasts are usually hard to reconstruct
via locally convergent algorithms. On the other hand, these values
are especially important for our target application, since they
might be used for monitoring stroke treatments. Therefore, following
the concept of Steps 1-6 of subsection 4.1, we conclude that our
approximate mathematical model of subsection 4.3 is a valid one.

\begin{center}
\textbf{Acknowledgments}
\end{center}

The work of all authors was supported by the National Institutes of
Health grant number 1R21NS052850-01A1. In addition, the work of MVK
was supported by the U.S. Army Research Laboratory and U.S. Army
Research Office under the grant number W911NF-11-1-0399.


\label{lastpage}

\end{document}